\documentclass{article}
\usepackage[utf8]{inputenc}

\usepackage{enumitem}
\usepackage{microtype}
\usepackage{graphicx}
\usepackage{subfigure}
\usepackage[dvipsnames]{xcolor}
\usepackage{booktabs} 
\usepackage{multirow}

\usepackage[margin=3cm]{geometry}

\usepackage{rotating}

\usepackage[colorlinks=true, citecolor=blue, urlcolor=black, bookmarks=false]{hyperref}

\usepackage{amsmath}
\usepackage{amssymb}
\usepackage{mathtools}
\usepackage{amsthm}
\usepackage[capitalize,noabbrev]{cleveref}
\usepackage{authblk}

\theoremstyle{plain}
\newtheorem{theorem}{Theorem}[section]
\newtheorem{proposition}[theorem]{Proposition}
\newtheorem{lemma}[theorem]{Lemma}

\theoremstyle{definition}
\newtheorem{definition}[theorem]{Definition}

\theoremstyle{remark}
\newtheorem{remark}[theorem]{Remark}

\crefname{section}{Section}{Sections}
\crefname{subsection}{Subsection}{Subsections}
\crefname{equation}{Equation}{Equations}
\crefname{figure}{Figure}{Figures}

\newcommand{\kmax}{\mathop{k\operatorname{-max}}\limits}
\usepackage[textsize=tiny]{todonotes}

\newcommand{\R}{\mathbb{R}}
\newcommand{\N}{\mathbb{N}}

\newcommand{\Z}{\mathbb{Z}}

\newcommand{\E}{\mathbb{E}}

\newcommand{\XX}{\mathcal{X}}

\newcommand{\DD}{\mathcal{D}}

\newcommand{\CC}{\mathcal{C}}

\newcommand{\OO}{\mathcal{O}}
\newcommand{\NN}{\mathcal{N}}

\newcommand{\bP}{\mathbf{P}}
\newcommand{\bQ}{\mathbf{Q}}

\renewcommand{\epsilon}{\varepsilon}

\newcommand{\dd}{\mathrm{d}}  

\newcommand{\defeq}{\vcentcolon=}
\newcommand{\eqdef}{=\vcentcolon}

\newcommand{\ReLU}{\mathrm{ReLU}}

\newcommand{\dgm}{\mathrm{Dgm}}

\DeclareMathOperator*{\op}{\mathbf{op}}

\newcommand{\gudhi}{\href{https://gudhi.inria.fr/}{\texttt{Gudhi}}}
\newcommand{\ucr}{\href{https://www.cs.ucr.edu/~eamonn/time_series_data_2018/}{\texttt{UCR}}}
\newcommand{\modelnet}[1]{\href{https://modelnet.cs.princeton.edu/}{\texttt{ModelNet{#1}}}}
\newcommand{\scikit}{\href{https://scikit-learn.org/stable/}{\texttt{Scikit-Learn}}}
\newcommand{\tensorflow}{\href{https://www.tensorflow.org/}{\texttt{TensorFlow}}}
\newcommand{\velour}{\href{https://pypi.org/project/velour/}{\texttt{Velour}}}

\newcommand{\noise}[1]{\widetilde{#1}}
\newcommand{\dtm}[1]{{#1}^{\rm DTM}}

\newcommand{\PI}{\mathrm{PI}}
\newcommand{\RipsNet}{\mathbf{RN}}
\newcommand{\DeepSet}{\mathbf{DS}}
\newcommand{\knn}{\mathbf{kNN}}
\newcommand{\train}{\mathrm{Tr}}
\newcommand{\test}{\mathrm{Te}}
\newcommand{\class}{\mathrm{Cl}}
\newcommand{\PCs}{\mathrm{PC}}
\newcommand{\PDs}{\mathrm{PD}}
\newcommand{\PV}{\mathrm{PV}}

\newcommand{\Ls}{\mathrm{L}}
\newcommand{\RR}{\mathcal{R}}

\begin{document}

\title{RipsNet: a general architecture for fast and robust estimation of the persistent homology of point clouds}

\author[1]{Thibault de Surrel\footnote{These authors contributed equally to the work.}}
\author[2]{Felix Hensel{$^*$}}
\author[1]{Mathieu Carri\`ere{$^*$}}
\author[3]{Th\'eo Lacombe}
\author[4]{Yuichi Ike}
\author[5]{Hiroaki Kurihara}
\author[2]{Marc Glisse}
\author[2]{Fr\'ed\'eric Chazal}

\affil[1]{Universit{\'e} C{\^o}te d'Azur, Inria, France}
\affil[2]{Universit{\'e} Paris-Saclay, CNRS, Inria, Laboratoire de Math{\'e}matiques d'Orsay, France}
\affil[3]{LIGM, Université Gustave Eiffel, Champs-sur-Marne, France}
\affil[4]{The University of Tokyo, Japan}
\affil[5]{Fujitsu Limited, Japan}

\maketitle

\begin{abstract}
The use of topological descriptors in modern machine learning applications, such as Persistence Diagrams (PDs) arising from Topological Data Analysis (TDA), has shown great potential in various domains.
However, their practical use in applications is often hindered by two major limitations: the computational complexity required to compute such descriptors exactly, and their sensitivity to even low-level proportions of outliers. In this work, we propose to bypass these two burdens in a data-driven setting by entrusting the estimation of (vectorization of) PDs built on top of point clouds to a neural network architecture that we call \emph{RipsNet}. Once trained on a given data set, RipsNet can estimate topological descriptors on test data very efficiently with generalization capacity. Furthermore, we prove that RipsNet is robust to input perturbations in terms of the 1-Wasserstein distance, a major improvement over the standard computation of PDs that only enjoys Hausdorff stability, yielding RipsNet to substantially outperform exactly-computed PDs in noisy settings. We showcase the use of RipsNet on both synthetic and real-world data.
Our open-source implementation is publicly available\footnote[2]{\url{https://github.com/hensel-f/ripsnet}} and will be included in the \gudhi\,library.
\end{abstract}

\section{Introduction}

The knowledge of topological features (such as connected components, loops, and higher dimensional cycles) that are present in data sets provides a better understanding of their structural properties at multiple scales, and can be leveraged to improve statistical inference and prediction.
Topological Data Analysis (TDA) is the branch of data science that aims to detect and encode such topological features, in the form of
\emph{persistence diagrams} (PD). 
A PD is a (multi-)set of points $D$ in $\R^2$, in which each point $p\in D$ corresponds to a topological feature of the data whose size is encoded by its coordinates.
PDs are descriptors of a general nature and allow flexibility in their computation.
As such, they have been successfully applied to many different areas of data science, including, e.g., material science~\cite{Buchet2018}, genomic data~\cite{Camara2017}, and 3D-shapes~\cite{Li2014}. 
In the present work, we focus on PDs stemming from point cloud data, 
referred to as \emph{Rips PDs}, 
which find natural use in shape analysis \cite{chazal2009gromov,gamble2010exploring} but also in other domains such as time series analysis \cite{Perea2015, pereira2015persistent, umeda2017time}, or in the study of the behavior of deep neural networks \cite{guss2018characterizing,naitzat2020topology,birdal2021intrinsic}. 

A drawback of Rips PDs computed on large point clouds is that they are computationally expensive.
Furthermore, even though these topological descriptors enjoy stability properties with respect to the input point cloud in the Hausdorff metric~\cite{chazal2014persistence}, they are fairly sensitive to perturbations: moving a single point in an arbitrarily large point cloud can alter the resulting Rips PD substantially.

In addition, the lack of linear structure (such as addition and scalar multiplication) of the space of PDs 
hinder the use of PDs in standard machine learning pipelines, which are typically developed to handle inputs belonging to a finite dimensional vector space.
This burden motivated the development of \emph{vectorization methods}, which allow to map PDs into a vector space while preserving their structure and interpretability.
Vectorization methods can be divided into two classes, {\em finite-dimensional embeddings}~\cite{Bubenik2015, Adams2017, Carriere2015, Chazal2015, Kalisnik2018},
turning PDs into elements of Euclidean space $\R^d$, and {\em kernels}~\cite{Carriere2017, Kusano2016, Le2018, Reininghaus2015}, 
that implicitly map PDs to elements of infinite-dimensional Hilbert spaces.


In this work, we propose to overcome the previous limitations of Rips PDs, by learning their finite-dimensional embeddings directly from the input point cloud data sets with neural network architectures. This approach allows not only for a much faster computation time, but also for increased robustness 
of the topological descriptors. 

\subsection*{Contributions}

More specifically, our contributions in this work are summarized as follows.
\begin{itemize}[topsep=0pt]
    \item We introduce RipsNet, a DeepSets-like architecture capable of learning finite-dimensional embeddings of Rips PDs built on top of point clouds. 
    \item We study the robustness properties of RipsNet. In particular, we prove that for a given point cloud $X$, perturbing a proportion $\lambda \in (0,1)$ of its points can only change the output of RipsNet by $\OO(\lambda)$, while the exact persistence diagram may change by a fixed positive quantity even in the regime $\lambda \to 0$. 
    \item We experimentally showcase how RipsNet can be trained to produce fast, accurate, and useful estimations of topological descriptors. In particular, we observe that using RipsNet outputs instead of exact PDs yields better performances for classification tasks based on topological properties. 
\end{itemize}

\subsection*{Related work}

\textbf{DeepSets.} Our RipsNet architecture is directly based on \emph{DeepSets} \cite{Zaheer2017}, a particular case of equivariant neural network \cite{cohen2021equivariant} designed to handle point clouds as inputs. Namely, DeepSets essentially consist of processing a point cloud
$X = \{x_1,\dots,x_n\} \subset \R^d$
via
\begin{equation} \label{eq:deepset}
X \mapsto \phi_2 \left( \op(\{\phi_1(x_i)\}_{i=1}^{n}) \right),
\end{equation}
where $\op$ is a \emph{permutation invariant operator} on sets (such as sum, mean, maximum, etc.) and $\phi_1 : \R^d \to \R^{d'}$ and $\phi_2 : \R^{d'} \to \R^{d''}$ are parametrized maps (typically encoded by neural networks) optimized in the training phase. Eq.~\eqref{eq:deepset} makes the output of DeepSets architectures invariant to permutations, a property of Rips PDs that we want to reproduce in RipsNet. 

\textbf{Learning to estimate PDs.} There exist a few works attempting to compute or estimate (vectorizations of) PDs through the use of neural networks. 
In \cite{som2020piNet}, the authors propose a convolutional neural network (CNN) architecture to estimate persistence images (see Section~\ref{sec:persvec}) computed on 2D-images. 
Similarly, in \cite{montufar2020canNNlearnPH}, the authors provide an experimental overview of specific PD features (such as, e.g.,~their tropical coordinates \cite{kalivsnik2019tropical}) that can be learned using a CNN, when PDs are computed on top of 2D-images. 
On the other hand, RipsNet is designed to handle the (arguably harder) situation where input data are point clouds of arbitrary cardinality instead of 2D-images (i.e.,~vectors). 
Finally, the recent work \cite{zhou2021learning} 
also aims at
learning to compute topological descriptors on top of point clouds via a neural network. 
However, note that our methodology is quite different: while our approach based on a DeepSet architecture allows to process point clouds directly, the approach proposed in \cite{zhou2021learning} requires the user to equip the point clouds with graph structures (that depend on hyper-parameters mimicking Rips filtrations). 
Furthermore, a key difference between our approach and the aforementioned works is that we provide a theoretical study of our model that provides insights on its behavior, particularly in terms of robustness to noise, while the other works are mostly experimental.

\section{Background}

\subsection{Persistent homology and Rips PDs}
In this section, we briefly recall the basics of ordinary persistence theory. We refer the interested reader to~\cite{Cohen-Steiner2009, Edelsbrunner2010, Oudot2015} for a thorough treatment.

\textbf{Persistence diagrams.}
Let $\XX$ be a topological space, and $f \colon \XX \to \R$ a real-valued continuous function. The $\alpha$-\emph{sublevel set} of $(\XX,f)$ is defined as 
$\XX_\alpha=\{x \in \XX \,:\, f(x) \leq \alpha \}$.
Increasing $\alpha$ from $-\infty$ to $+\infty$ yields an increasing nested sequence of sublevel sets, called the {\em filtration} induced by $f$. 
It starts with the empty set and ends with the entire space $\XX$. 
Ordinary persistence keeps track of the times of appearance and disappearance of topological features (connected components, loops, cavities, etc.) in this sequence.
For instance, one can store the value $\alpha_b$, called the \emph{birth time}, for which a new connected component appears in $X_{\alpha_b}$. 
This connected component eventually merges
with another one for some value $\alpha_d \geq \alpha_b$, which is stored as well and called the \emph{death time}.
One says that the component \emph{persists} on the corresponding interval $[\alpha_b,\alpha_d]$.
Similarly, we save the $[\alpha_b,\alpha_d]$ values of each loop, cavity, etc.~that appears in a specific sublevel set $X_{\alpha_b}$ and disappears (gets ``filled'') in $X_{\alpha_d}$. 
This family of intervals is called the barcode, or \emph{persistence diagram}, of $(\XX,f)$, and can be represented as a multiset (i.e., elements are counted with multiplicity) of points supported on the open half-plane $\{ (\alpha_b,\alpha_d) \in \R^2:\ \alpha_b < \alpha_d \}\subset \R^2$. 
The information of connected components, loops, and cavities is represented in persistence diagrams of dimension $0$, $1$, and $2$, respectively.

\textbf{Filtrations for point clouds.} Let $X=\{x_{1},\dots,x_{n}\}$ be a finite point cloud in $\XX = \mathbb{R}^d$, and $\alpha \geq 0$. 
Let $f_{X} : v \mapsto \min_{x \in X} \|v-x\|$ denote the distance of $v \in \R^d$ to $X$. 
In this case, the sublevel set $\XX_{\alpha}$ is given by the union of $d$-dimensional closed balls of radius $\alpha$ centered at $x_{i}$ ($i=1,\dots,n$).
From this filtration, different types of PDs can be built, called \v Cech, Rips, and alpha filtrations, which can be considered as being equivalent for the purpose of this work (in particular, RipsNet can be used seamlessly with any of these choices)---we refer to Appendix~\ref{app:filtrations} for details. 
Due to its computational efficiency in low-dimensional settings, we use the alpha filtration in our numerical experiments.


\textbf{Metrics between persistence diagrams.} The space of persistence diagrams can be equipped with a parametrized metric $d_s$, $1 \leqslant s \leqslant \infty$ which is rooted in algebraic considerations and whose proper definition is not required in this work. 
In the particular case $s = \infty$, this metric will be referred to as the \emph{bottleneck} distance between persistence diagrams.
Of importance is, that the space of persistence diagrams $\DD$ equipped with such metrics lacks linear (Hilbert; Euclidean) structure \cite{carriere2018metric, bubenik2020embeddings}. 

\subsection{Vectorizations of persistence diagrams}\label{sec:persvec}

The lack of linear structure of the metric space $(\DD, d_s)$ prevents a faithful use of persistence diagrams in standard machine learning pipelines, as such techniques typically require inputs belonging to a finite-dimensional vector space.
A natural workaround is thus to seek for a \emph{vectorization} of persistence diagrams ($\PV$), that is a map $\phi : (\DD, d_s) \to (\R^d, \| \cdot \|)$ for some dimension $d$. 
Provided the map $\phi$ satisfies suitable properties (e.g.,~being Lipschitz, injective, etc.), one can turn a sample of diagrams
$\{D_1,\dots, D_n\} \subset \DD$
into a collection of vectors
$\{\phi(D_1),\dots,\phi(D_n)\} \subset \R^d$
which can subsequently be used 
to perform any 
machine learning task.

Various vectorizations techniques, with success in applications, have been proposed
\cite{Carriere2015, Chazal2015, Kalisnik2018}. 
In this work, we focus, for the sake of concision, on two of them: the \emph{persistence image} (PI)~\cite{Adams2017} and the \emph{persistence landscape} (PL)~\cite{Bubenik2015}---though the approach developed in this work adapts faithfully to any other 
vectorization. 

\textbf{Persistence images (PI).} Given a persistence diagram $D$, computing its persistence image essentially boils down to putting a Gaussian
\[
    g_u(z) \coloneqq \frac{1}{2\pi\sigma^2} \exp\left( -\frac{\|z-u\|^2}{2\sigma^2} \right),
\]
with fixed variance $\sigma^2$, on each of its points $u$ and weighing it by a piecewise differentiable function $w : \R^2 \to \R_{\geq 0}$ (typically a function of the distance of $u$ to the diagonal $\{(t,t)\} \subset \R^2$) and then discretizing the resulting surface on a fixed grid to obtain an image. 
Formally, one starts by rotating the diagram $D$ via the map $T \colon \R^2 \to \R^2, (b,d) \mapsto (b,d-b)$.
The \emph{persistence surface} of $D$ is defined as 
\[
    \rho_D(z) 
    \coloneqq 
    \sum_{u \in T(D)} w(u) g_u(z),
\]
where $w$ satisfies $w(x,0)=0$.
Now, given a compact subset $A \subset \R^2$ partitioned into domains  $A=\bigsqcup_{i=1}^k P_i$---in practice a rectangular grid regularly partitioned in $(n \times n)$ pixels---we set $I(\rho_D)_P \coloneqq \int_P \rho_D dz$.
The vector $(I(\rho_D)_{P_i})_{i=1}^n$ is the persistence image of $D$. 
The transformation $\PI \colon X \mapsto \dgm(X) \mapsto I(\rho_{\dgm(X)})$ defines a finite-dimensional vectorization as illustrated in \cref{fig:pipeline_pi}.


\begin{figure}
    \centering
    \includegraphics[width=\columnwidth]{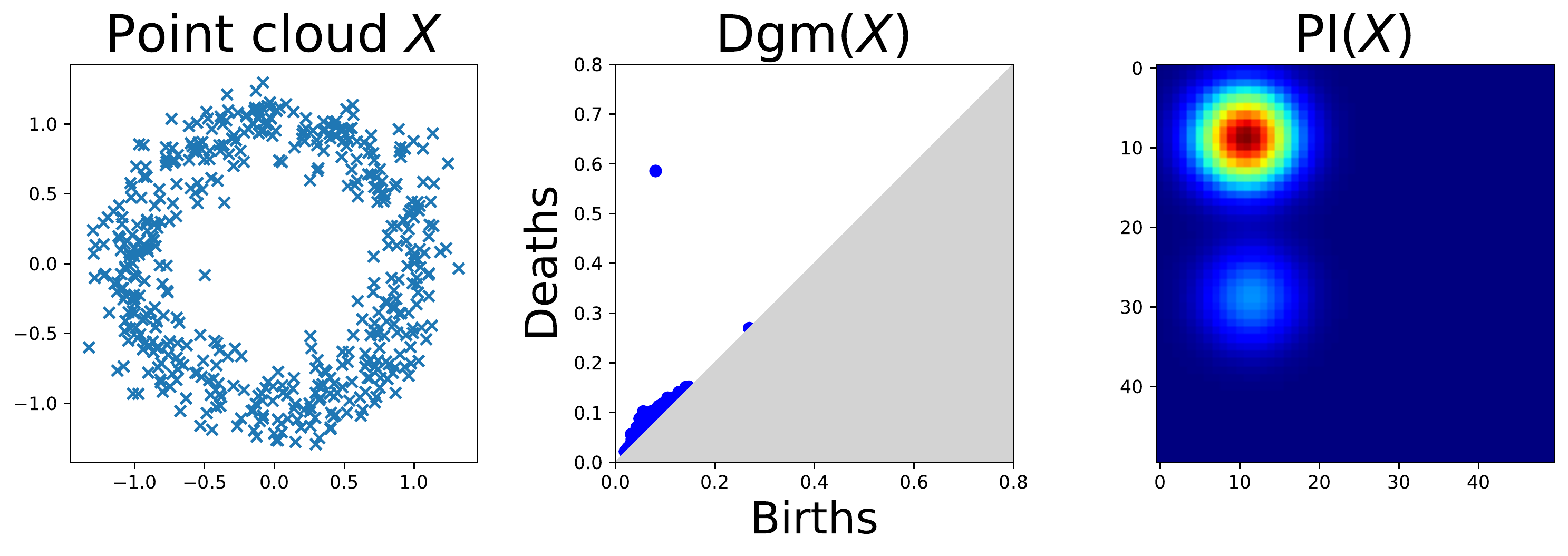}
    \caption{
    Pipeline to extract persistence images from point clouds. The prominent bump in $\PI(X)$ corresponds to the distinguished point in $\dgm(X)$ accounting for the circle underlying the point cloud $X$. The smaller bump accounts for smaller loops inferred when filtering $X$, which yield points near the diagonal in $\dgm(X)$ that are downweighted by $w(x,y)=\tanh(y)$ in $\PI(X)$.}
    \label{fig:pipeline_pi}
\end{figure}

\textbf{Persistence landscapes (PL).} \cite{Bubenik2015} Given a persistence diagram $D=\{ (b_i,d_i) \}_{i=1}^n$, we define a function $\lambda_k(D) \colon \R \to \R_{\ge 0}$ by
\[
\lambda_k(D)(t)
    =
    \kmax_i \min\{t-b_i,d_i-t\}_+,
\]
for each $k\in\Z_{\geq 1}$, where $\kmax$ denotes the $k$-th largest value in the set or $0$ if the set contains less than $k$ points, and $a_+=\max\{0,a\}$ for a real number $a$.
The sequence of functions $(\lambda_k(D))_{k=1}^\infty$ is called the \emph{persistence landscape} of the persistence diagram $D$. In practice, these landscape functions are evaluated on a 1-D grid, and the corresponding values are concatenated into a vector.

\section{RipsNet}

\subsection{Motivation and definition}
In the pipeline illustrated in \cref{fig:pipeline_pi}, the computation of the persistence diagram $\dgm(X)$ from an input point cloud $X$ is the most complex operation involved: it is both computationally expensive and introduces non-differentiability in the pipeline. 
Moreover, as detailed in \cref{subsec:instability_std_rips} for the specific case of persistence images, the output vectorization can be highly sensitive to perturbations in the input point cloud $X$: moving a single point $p_i \in X$ can arbitrarily change $\PI(X)$ even when $n$ is large. 
This instability can be a major limitation when incorporating persistence vectorizations ($\PV$s) of diagrams in practical applications. 


To overcome these difficulties, we propose a way to bypass this computation by designing a neural network architecture which we call \emph{RipsNet} ($\RipsNet$).
The goal is to learn a function, denoted by $\RipsNet$ as well, able to reproduce persistence vectorizations for a given distribution of input point clouds $X \sim \bP$ after being trained on a sample $\{X_i\}_{i=1}^n$
with labels being the corresponding vectorizations
$\{\PV(X_i)\}_{i=1}^n$.

As $\RipsNet$ takes point clouds $X = \{x_1,\dots,x_N\} \subset \R^d$ of potentially varying sizes as input, it is natural to expect it to be \emph{permutation invariant}. 
An efficient way to enforce this property is to rely on a DeepSets architecture \cite{Zaheer2017}. 
Namely, it consists of decomposing the network into a succession of two maps $\phi_1 \colon \R^d \to \R^{d'}$ and $\phi_2 \colon \R^{d'} \to \R^{d''}$ and a permutation invariant operator $\op$---typically the sum, the mean, or the maximum
\[
    \RipsNet \colon X \mapsto \phi_2 \left( \op(\{  \phi_1(x)\}_{x\in X})  \right).
\]
For each $x \in X$, the map $\phi_1$ provides a representation $\phi_1(x)$; these point-wise representations are gathered via the permutation invariant operator $\op$, and the map $\phi_2$ is subsequently applied to compute the network output. 

In practice, $\phi_1$ and $\phi_2$ are themselves parameterized by neural networks; in this work, we will consider simple feed-forward fully-connected networks (see \cref{sec:expe} for the architecture hyper-parameters), though more general architectures could be considered. 
The parameters characterizing $\phi_1$ and $\phi_2$ are tuned during the training phase, where we minimize the $L2$-loss
\begin{equation}\label{eq:objective_function}
    \sum_{i=1}^n \| \RipsNet(X_i) - \PV(X_i)\|^2,
\end{equation}
over a set of training point clouds $\{X_i\}_i$ with corresponding pre-computed vectorizations $\{\PV(X_i)\}_i$. 

Once trained properly (assuming good generalization properties), when extracting topological information of a point cloud, an important advantage of using $\RipsNet$ instead of $\PV$ lies in the computational efficiency: while the exact computation of persistence diagrams and vectorizations rely on expensive combinatorial computations, running the forward pass of a trained network is significantly faster, as showcased in \cref{sec:expe}. 
As detailed in the following subsection and illustrated in our experiments, $\RipsNet$ also satisfies some strong robustness properties, yielding a substantial advantage over exact $\PV$s when the data contain some perturbations such as noise, outliers, or adversarial attacks.

\subsection{Wasserstein stability of RipsNet}
In this subsection, we show that RipsNet satisfies robustness properties.
A convenient formalism to demonstrate these properties is to represent a point cloud $X = \{x_1,\dots,x_N\}$ by a probability measure $m_X \defeq \frac{1}{N} \sum_{i=1}^N \delta_{x_i}$, where $\delta_{x_i}$ denotes the Dirac mass located at $x_i \in \R^d$. 
Let $\mu(f)$ denote $\int f \dd \mu$ for a map $f$ and a probability measure $\mu$. 
Such measures can be compared by Wasserstein distances $W_p,\ p \geq 1$, which are defined for any two probability measures $\mu,\nu$ supported on a compact subset $\Omega \subset \R^d$ as
\[
    W_p(\mu,\nu) \defeq \left(\inf_{\pi} \iint \|x-y\|^p \dd \pi(x,y) \right)^{\frac{1}{p}},
\]
where the infimum is taken over measures $\pi$, supported on $\R^d \times \R^d$, with marginals $\mu$ and $\nu$. 
We also mention the so-called Kantorovich--Rubinstein duality formula that occurs when $p=1$:
\begin{equation}\label{eq:duality}
    W_1(\mu,\nu) = \sup_{f:\text{$1$-Lip}} \left(\int f \dd\mu - \int f \dd\nu\right).
\end{equation}

Throughout this section, we fix $\op$ as the mean operator:  $\op(\{y_1,\dots,y_N\}) = \frac{1}{N}\sum_{i=1}^N y_i$, and let $\RipsNet = \phi_2 \circ \op \circ \phi_1$, where $\phi_1,\phi_2$ are two Lipschitz-continuous maps with Lipschitz constant $C_1,C_2$, respectively.

\subsubsection{Pointwise stability}
If $X = \{x_1,\dots, x_{N-1}, x_N\} \subset \Omega$ and $X' = \{x_1,\dots, x_{N-1}, x'_N\} \subset \Omega$, it is worth noting that $W_1(m_X, m_{X'}) \leq \frac{1}{N} \|x_N - x'_N\|$. 
Therefore, moving a single point $x_N$ of $X$ to another location $x'_N$ changes the $W_1$ distance between the two measures by at most $\OO(1/N)$. 
More generally, moving a fraction $\lambda \in (0,1)$ of the points in $X$ affects the Wasserstein distance in $\OO(\lambda)$.
$\RipsNet$ satisfies the following stability result.

\begin{proposition}
\label{prop:ripsnet_pw_stability}
For any two point clouds $X,Y$, and any $p \geq 1$, one has 
\begin{align*} 
\| \RipsNet(X) - \RipsNet(Y) \| &\leq C_1 C_2 \cdot W_1(m_X, m_Y) \\
&\leq C_1 C_2 \cdot W_p(m_X,m_Y).
\end{align*}
\end{proposition}

\begin{proof}
We have 
\begin{align*}
    \| \RipsNet(X) - \RipsNet(Y) \| & = \| \phi_2(m_X(\phi_1)) - \phi_2(m_Y(\phi_1)) \| \\
    & \leq C_2 \| m_X(\phi_1) - m_Y(\phi_1) \| \\
    & \leq C_2 C_1 \sup_{f: \text{$1$-Lip}} \| m_X(f) - m_Y(f)\| \\
    & \leq C_2 C_1 W_1(m_X,m_Y),
\end{align*}
where we used \eqref{eq:duality} and we conclude $W_1 \leq W_p$ using Jensen's inequality. 
\end{proof}

In particular, this result implies that moving a small proportion of points $\lambda$ in $\Omega$ in a point cloud $X$ does not affect the output of $\RipsNet$ by much. 
We refer to it as a ``pointwise stability'' result in the sense that it describes how $\RipsNet$ is affected by perturbations of a fixed point cloud $X$. 

Note that in contrast, Rips PDs, as well as their vectorizations, are not robust to such perturbations: moving a single point of $X$, even in the regime $\lambda \to 0$, may change the resulting persistence diagram by a fixed positive amount, preventing a similar result to hold for $\PV$s. 
A concrete example of this phenomenon is given in \cref{subsec:instability_std_rips} for the case of persistence images. 

\subsubsection{Probabilistic stability}
\label{subsec:noise}
The pointwise stability result of \cref{prop:ripsnet_pw_stability} can be used to obtain a good theoretical understanding on how RipsNet behaves in practical learning settings. 
For this, we consider the following model: let $P$ be a law on some compact set $\Omega \subset \R^d$, fix $N \in \N$, and let $\bP$ denote $P^{\otimes N}$, that is, $X \sim \bP$ is a random point cloud $X = \{x_1,\dots, x_N\}$ where the $x_i$'s are i.i.d. $\sim P$.

In practice, given a training sample $X_1,\dots,X_n \sim \bP$, RipsNet is trained to minimize the empirical risk 
\[
    \widehat{\RR}_n \defeq \frac{1}{n} \sum_{i=1}^n \| \RipsNet(X_i) - \PV(X_i)\|,
\]
which, hopefully, yields a small theoretical risk:
\[ 
    \RR \defeq \int \| \RipsNet(X) - \PV(X)\| \dd \bP(X).
\]

\begin{remark}
The question to know whether ``$\widehat{\RR}_n$ small $\Rightarrow$ $\RR$ small'' is related to the capacity of RipsNet to generalize properly. 
Providing a theoretical setting where such an implication should hold is out of the scope of this work, but can be checked empirically by looking at the performances of RipsNet on validation sets.
\end{remark}

We now consider the following noise model: given a point cloud $X \sim \bP$, we randomly replace a fraction $\lambda = \frac{N-K}{N} \in (0,1)$ of its points\footnote{As the $x_i$'s are i.i.d., we may assume without loss of generality that the last $N-K$ points are replaced.} by corrupted observations distributed with respect to some law $Q$. 
Let $Y \sim Q^{\otimes N-K} \eqdef \bQ$ and $F(X,Y)$ denote this corrupted point cloud.

\begin{lemma}
\label{lemma:kopkop}
Let $C(P,Q) \defeq \E_{P \otimes Q}[\|x - y\|]$. Then,
\[ \E_{\bP \otimes \bQ} [W_1(F(X,Y), X)] \leq \lambda C(P,Q). \]
In particular, if $P,Q$ are supported on a compact set $\Omega \subset \R^d$ with diameter $\leq L$, the bound becomes $\lambda L$. 
\end{lemma}

\begin{proof}
Set $X = \{x_1,\dots,x_N\}$ and $Y = \{y_{K+1},\dots,y_N\}$. 
Assume without loss of generality that $F(X,Y) = \{x_1,\dots, x_K, y_{K+1},\dots, y_N\}$, where $K = (1-\lambda) N$. 
Let us consider the transport plan that does not move the first $K$ points, and transports $V = \{ y_{K+1},\dots, y_N \}$ toward $ U = \{x_{K+1},\dots, x_N\}$ using the coupling $x_i \leftrightarrow y_i$. 
As this transport plan is sub-optimal, we have
\[ W_1(F(X,Y), X) \leq \frac{1}{N} \sum_{i=K+1}^N \|x_i - y_i\|. \]

Hence
\begin{align*} 
    & \E_{\bP \otimes \bQ}[W_1(F(X,Y),X)] \\
    = {} &\int W_1(F(X,Y),X) \dd \bP(X) \dd \bQ(Y) \\
    \leq {} &\frac{1}{N} \int \sum_{i=K+1}^N \|x_i - y_i\| \dd P(x_i) \dd Q(y_i) \\
    \leq {} & \frac{1}{N} \sum_{i=K+1}^N \int \|x_i - y_i\| \dd P(x_i) \dd Q(y_i) \\
    \leq {} & \frac{1}{N} \sum_{i=K+1}^N \E_{P \otimes Q}\|x - y\| \\
    \leq {} &\lambda \cdot C(P,Q),
\end{align*}
as claimed.
\end{proof}

We can now state the main result of this section.

\begin{proposition}\label{prop:main}
One has
\begin{align*}
    &\int \| \RipsNet(F(X,Y)) - \PV(X) \| \dd \bP(X) \dd \bQ(Y) \\
    \leq {} &\lambda C_1 C_2 \cdot C(P,Q) + \RR. 
\end{align*}
In particular, if $P,Q$ are supported on a compact subset of $\R^d$ with diameter $\leq L$, one has 
\[ \int \| \RipsNet(F(X,Y)) - \PV(X) \| \dd \bP(X) \dd \bQ(Y) \leq \OO(\lambda + \RR). \]
\end{proposition}

\begin{proof}
By \cref{prop:ripsnet_pw_stability}, we have
\begin{align*}
    &\int \| \RipsNet(F(X,Y)) - \PV(X) \| \dd \bP(X) \dd \bQ(Y) \\ 
    \leq {} &  \int \| \RipsNet(F(X,Y)) - \RipsNet(X) \| \dd \bP(X) \dd \bQ(Y) \\
    & \quad + \int \| \RipsNet(X) - \PV(X) \| \dd \bP(X) \\
    \leq {} & C_1 C_2 \int W_1(F(X,Y), X) \dd \bP(X) \dd \bQ(Y) + \RR \\
    \leq {} & C_1 C_2 \E_{\bP \otimes \bQ}[W_1(F(X,Y), X)] + \RR
\end{align*}
and we conclude using \cref{lemma:kopkop}.
\end{proof}

Therefore, if RipsNet achieves a low theoretical test risk ($\RR$ small) and only a small proportion $\lambda$ of points is corrupted, RipsNet will produce outputs similar, in expectation, to the persistence vectorizations $\PV(X)$ of the clean point cloud.

\subsection{Instability of standard Rips persistence images} \label{subsec:instability_std_rips}
Here we show that persistence images built on Rips diagrams do not satisfy a similar stability result. 
Namely, the idea is to replace the ``estimator'' $\RipsNet$ in the above section by the exact oracle $\PI$, for which $\RR = 0$, and to prove that
\[ 
    \int \| \PI(F(X,Y)) - \PI(X)\| \dd \bP \dd \bQ \not\to 0 
\]
in the regime $\lambda \to 0$ for some choice of underlying measures $P, Q$.

We consider the following setting:
\begin{itemize}[topsep=0pt]
    \item Let $P$ be the uniform distribution on a circle in $\R^2$ centered at $0$ with radius $1$. 
    \item Let $Q$ be the Dirac mass on $0$. 
    \item Fix $K=N-1$, that is, we move a single point of $X \sim \bP = P^{\otimes N}$ to $\bQ = Q = \delta_0$,  hence $\lambda = 1/N$.
    \item We consider persistence diagrams of dimension $1$, which represent loops in point clouds, and fix the variance $\sigma^2$ of the Gaussian used for the $\PI$ construction. 
\end{itemize}


In the regime $\lambda \to 0$, that is, $N \to \infty$, we have 
\[
    \PI(X) \to \NN((0,1),\sigma^2) \eqdef g_1
\]
almost surely. 
On the other hand, we have 
\[
    \PI(F(X,Y)) \to \NN\left(\left(0,1/2\right),\sigma^2 \right) \eqdef g_2
\]
almost surely, hence
\[ 
    \int \| \PI(F(X,Y)) - \PI(X)\| \dd \bP \dd \bQ \to \| g_1 - g_2 \|_1 > 0, 
\]
which proves the claim.


\section{Numerical experiments}
\label{sec:expe}
In this section, we illustrate the properties of our general architecture $\RipsNet$ presented in the previous sections. The approach we use is the following: we first train an $\RipsNet$ architecture on a training data set $\train_1$, comprised of point clouds $\PCs_1$ with
their corresponding labels $\Ls_1$ and persistence vectorizations $\PV_1$. Note that this training step does not require the labels $\Ls_1$ of the point clouds since the targets are the persistence vectorizations $\PV_1$. Then, we use both $\RipsNet$ and $\gudhi$ to compute the persistence vectorizations of three data sets: a second training data set $\train_2=(\PCs_2,\Ls_2)$, a test data set $\test=(\PCs,\Ls)$, and a noisy test data set $\noise{\test}=(\noise{\PCs}, \noise{\Ls})$ (as per our noise model explained in Section~\ref{subsec:noise}). All three data sets are comprised of labeled point clouds only.  At this stage, we also measure the computation time of $\RipsNet$ and $\gudhi$ for generating these persistence vectorizations.

In order to obtain quantitative scores, we finally train two machine learning classifiers: one on the labeled $\RipsNet$ persistence diagrams from $\train_2$, and the other on the labeled $\gudhi$ persistence diagrams from $\train_2$, which we call $\class_{\RipsNet}$ and $\class_{\gudhi}$, respectively.
The classifiers $\class_{\RipsNet}$ and $\class_{\gudhi}$ are then evaluated on the test persistence diagrams computed with $\RipsNet$ and $\gudhi$, respectively, on both $\test$ and $\noise{\test}$.
See Figure~\ref{fig:expe_method} for a schematic overview.

Finally, we also generate scores using AlphaDTM-based filtrations~\cite{Anai2019} computed with $\gudhi$ and the Python package~$\velour$ with parameters $m=5\%$ (respectively $m=0.75\%$ for the 3D-shape experiments), $p=2$, in the exact same way as we did for $\gudhi$. We let $\dtm{\gudhi}$ and $\class_{\dtm{\gudhi}}$ denote the corresponding model and classifier, respectively. 
Note that AlphaDTM-based filtrations usually require manual tuning, which, contrary to $\RipsNet$ parameters, cannot be optimized during training. 
In our experiments, we manually designed those parameters so that they provide reasonably-looking persistence vectorizations.
Finally, note that we also added some (non-topological) baselines in each of our experiments to provide a sense of what other methods
are capable of.
However, our main purpose is to show that $\RipsNet$ can provide a much faster and more noise-robust alternative to 
$\gudhi$, and hence the comparison we are most interested in is between $\RipsNet$ and both $\gudhi$ and $\dtm{\gudhi}$.

\begin{figure}[t]
    \centering
    \includegraphics[width=0.75\textwidth]{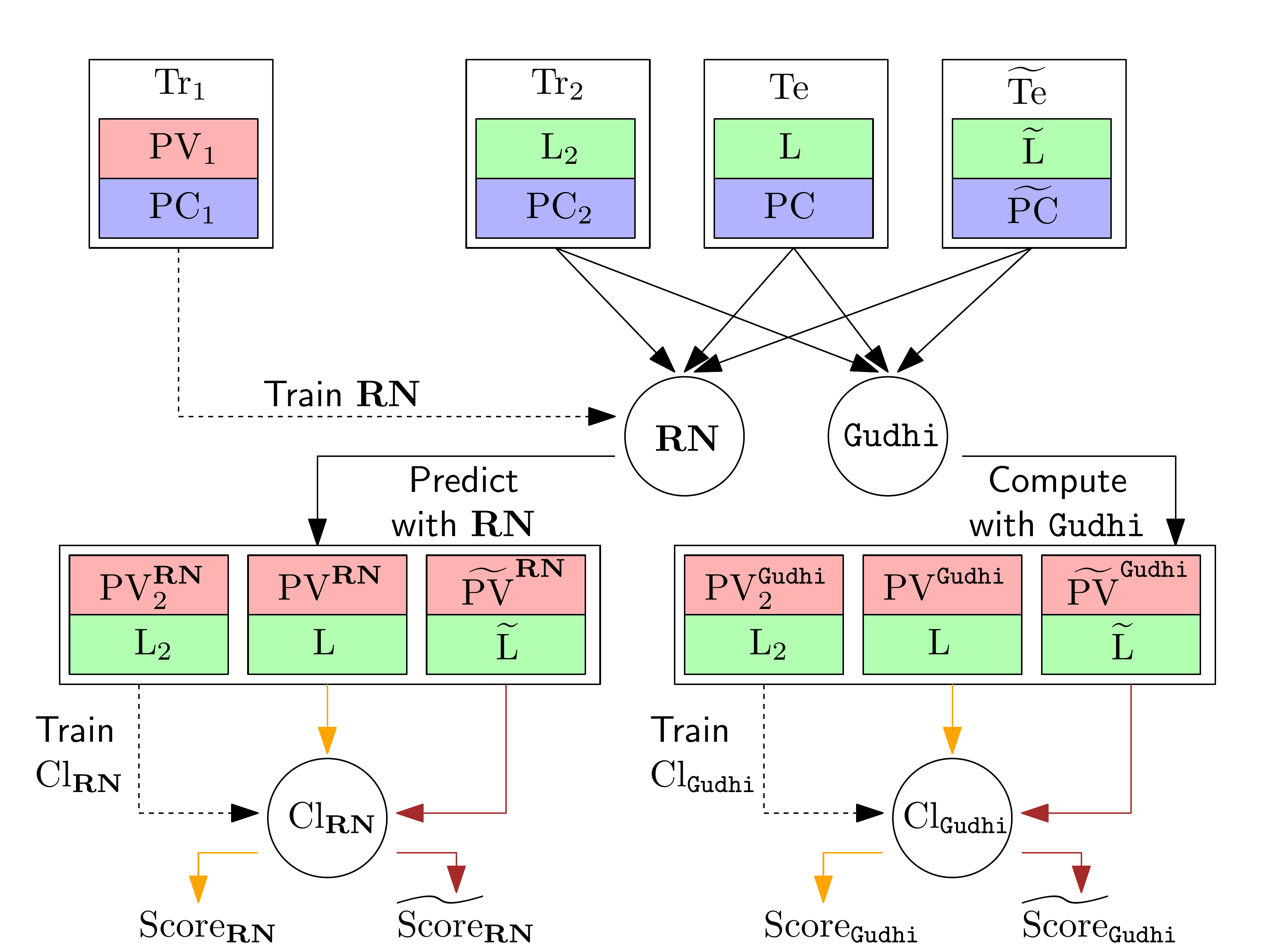}
    \caption{Scheme of our general experimental setup.}
    \label{fig:expe_method}
\end{figure}

In the following, we apply our experimental setup to three types of data.
First, we focus on synthetic data generated by a simple generative model.
Next, we consider time series data obtained from the $\ucr$ archive~\cite{UCR2018}, as well as 3D-shape data from the \modelnet{10} data set \cite{modelnet2015}.
The computations were run on a computing cluster on $4$ Xeon SP Gold $2.6$GHz CPU cores with $8$GB of RAM per core.

\subsection{Synthetic data}
\label{sec:synth}

\textbf{Dataset.} Our synthetic data set consists of samplings of unions of circles in the Euclidean plane $\R^2$. These unions are made of either one, two, or three circles, and we use the number of circles as the labels of the point clouds. Each point cloud has $N=600$ points, and $N-K=200$ corrupted points, i.e., $\lambda=1/3$, when noise is added. 

We train a RipsNet architecture $\RipsNet_{\rm synth}$ on a data set $\train_1$ of $3300$ point clouds, using $3000$ point clouds for training and $300$ as a validation set. The persistence diagrams $\PDs_1$ were computed with Alpha filtration in dimension $1$ with $\gudhi$, and then vectorized into either the first $5$ normalized persistence landscapes of resolution $300$ each, leading to $1500$-dimensional vectors, or into normalized persistence images of resolution $50\times 50$, leading to $2500$-dimensional vectors. 
The hyperparameters of these vectorizations were estimated from the corresponding persistence diagrams: the landscape limits were computed as the min and max of the $x$ and $y$ coordinates of the persistence diagrams points, while the image limits were computed as the min and max of the $x$ and $y-x$ coordinates, respectively.
Moreover, the image bandwidth was estimated as the $0.2$-quantile of all pairwise distances between the birth-persistence transforms of the persistence diagram points, and the image weight was defined as $10\cdot {\rm tanh}(y-x)$. 

Our architecture $\RipsNet_{\rm synth}$ is structured as follows.
The permutation invariant operator is $\op=$ sum, $\phi_1$ is made up of three fully connected layers of $30$, $20$, and $10$ neurons with $\ReLU$ activations, $\phi_2$ consists of three fully connected layers of $50$, $100$, and $200$ neurons with $\ReLU$ activations, and a last layer with sigmoid activation. 
We used the mean squared error (MSE) loss with Adamax optimizer with $\epsilon=5\cdot 10^{-4}$, and early stopping after $200$ epochs with less than $10^{-5}$ improvement. 

Finally, we evaluate $\RipsNet_{\rm synth}$ and $\gudhi$ using default XGBoost classifiers $\class_{\RipsNet_{\rm synth}}^{\rm XGB}$ and 
$\class_{\gudhi}^{\rm XGB}$ from $\scikit$, trained on a data set $\train_2$ of $3000$ point clouds and tested on a clean test set $\test$ and a noisy test set $\noise{\test}$ of $300$ point clouds each.
In addition, we compare it against two DeepSets architecture baselines trained directly on the point clouds: one ($\DeepSet_1$) with four fully connected layers of $50$, $30$, $10$, and $3$ neurons, and a simpler one ($\DeepSet_2$) with just two fully connected layers of $50$ and $3$ neurons. Both the architectures have $\ReLU$ activations, except for the last layer, permutation invariant operator $\op=$ sum, default Adam optimizer, and cross entropy loss from $\tensorflow$, and early stopping after $200$ epochs with less than $10^{-4}$ improvement.

\textbf{Results.} We show some point clouds of $\test$ and $\noise{\test}$, as well as their corresponding vectorized persistence diagrams and estimated vectorizations with $\RipsNet_{\rm synth}$, in Figure~\ref{fig:synth}. 
Accuracies and running times (averaged over 10 runs) are given in Tables~\ref{tab:results} and~\ref{tab:synth_time}.

\begin{figure}[t]
    \centering
    \includegraphics[width=0.19\textwidth]{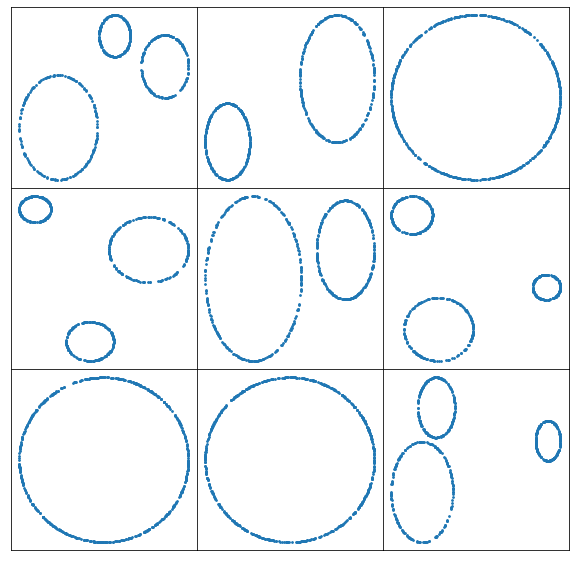}
    \includegraphics[width=0.19\textwidth]{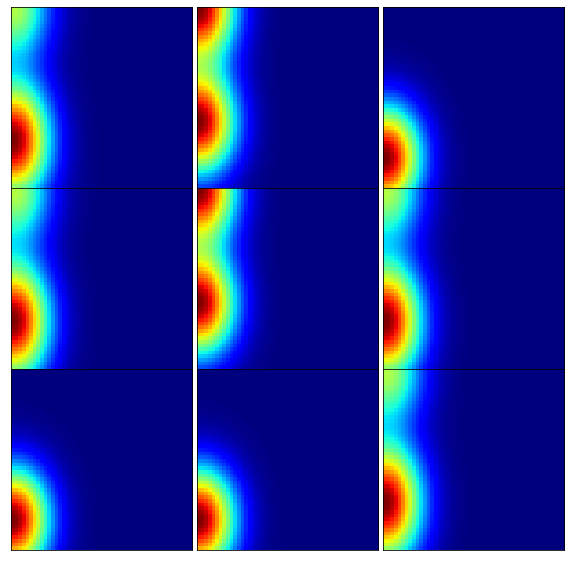}
    \includegraphics[width=0.19\textwidth]{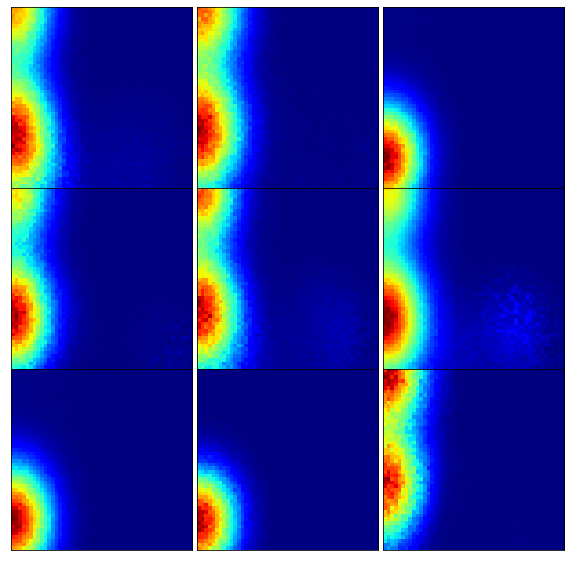}
    \includegraphics[width=0.19\textwidth]{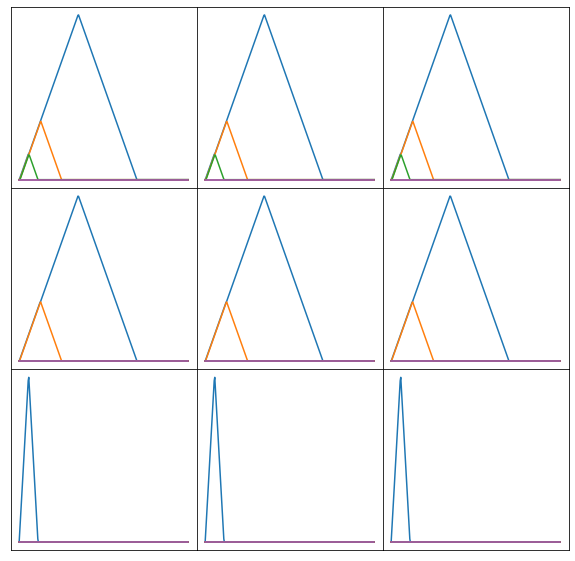}
    \includegraphics[width=0.19\textwidth]{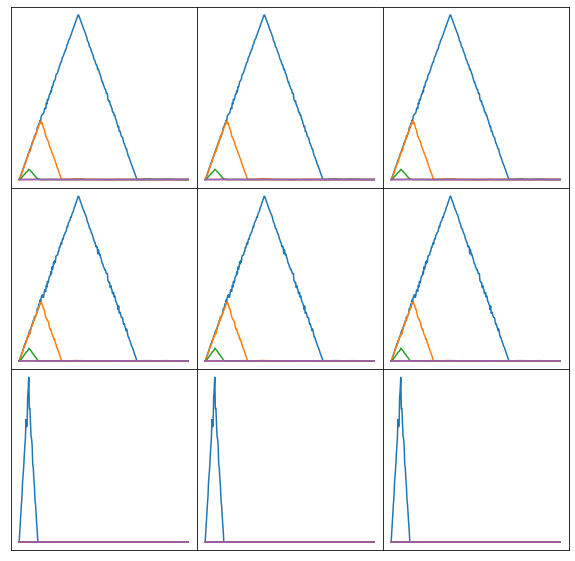}
    \includegraphics[width=0.19\textwidth]{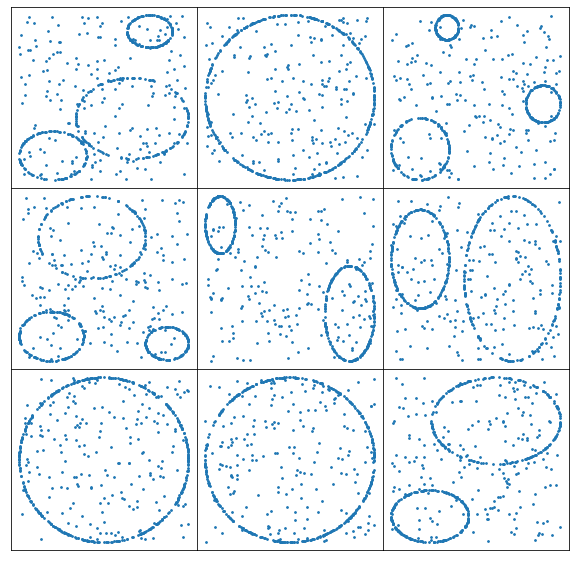}
    \includegraphics[width=0.19\textwidth]{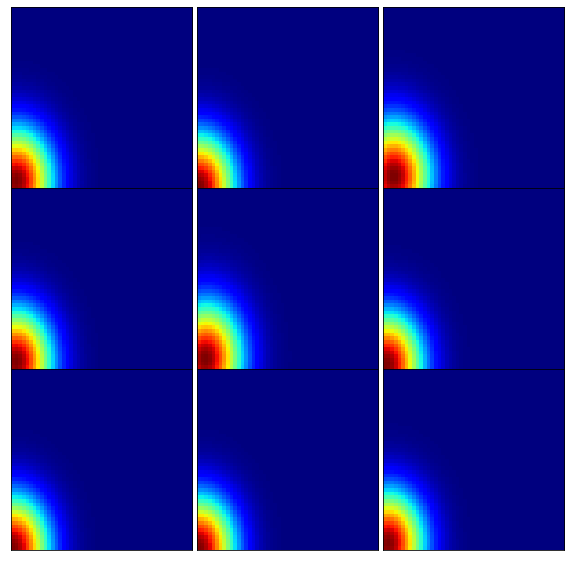}
    \includegraphics[width=0.19\textwidth]{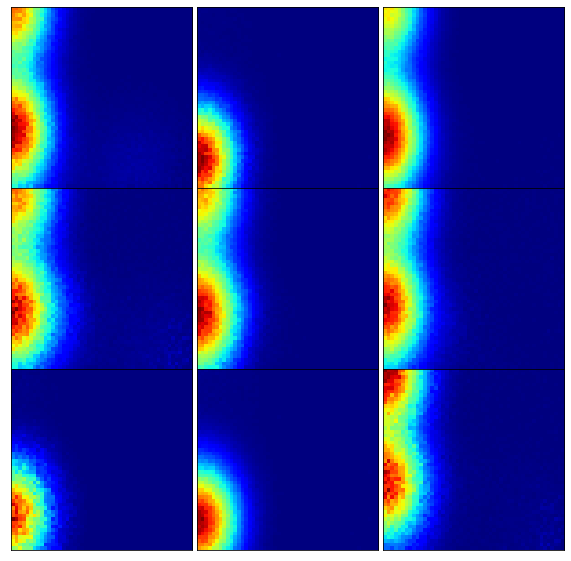}
    \includegraphics[width=0.19\textwidth]{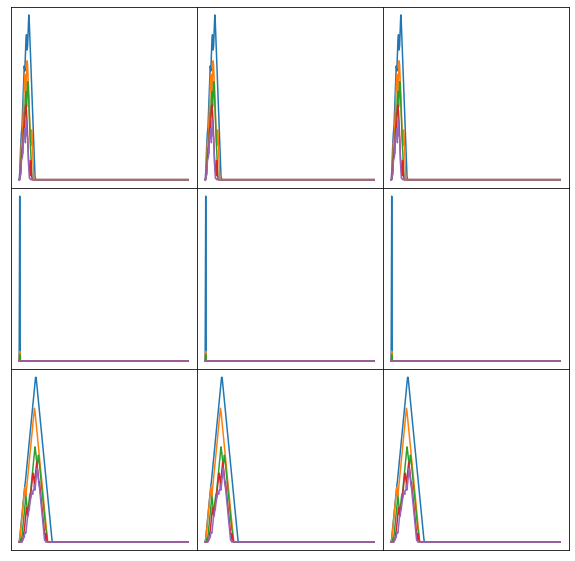}
    \includegraphics[width=0.19\textwidth]{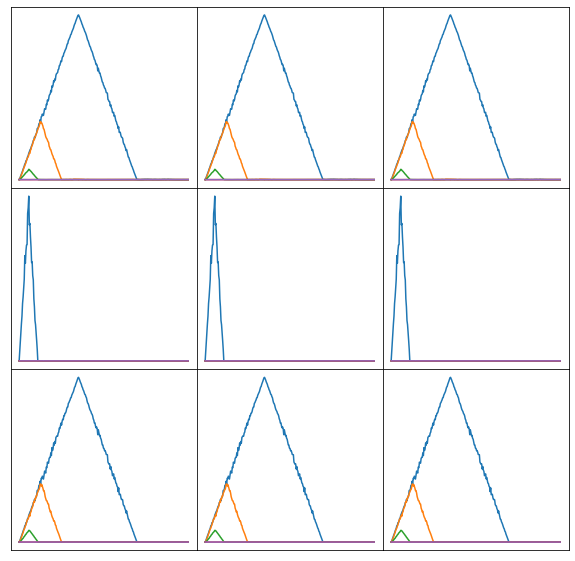}
    \caption{Point clouds (left), $\gudhi$ vectorizations (2nd and 4th columns) and $\RipsNet_{\rm synth}$ vectorizations (3rd and 5th columns) on clean data (1st row) and noisy data (2nd row) for PLs (4th and 5th columns) and PIs (2nd and 3rd columns).
    Larger version available in the Appendix (\cref{fig:synth_large}).
    }
    \label{fig:synth}
\end{figure}

\begin{table*}[]
    \centering
    \begin{tabular}{|c|ccc|cc|}
    \hline
        Synth. Data                & $\class_{\gudhi}^{\rm XGB}$ & $\class_{\dtm{\gudhi}}^{\rm XGB}$ & $\class_{\RipsNet}^{\rm XGB}$ & $\DeepSet_1$ & $\DeepSet_2$ \\                           
    \hline
    $\texttt{LS}$ & $99.9 \pm 0.1$ & $99.9 \pm 0.1$ & $80.7 \pm 3.0$ & $66.4 \pm 2.3$ & $66.0 \pm 2.4$ \\ 
$\texttt{PI}$ & $\color{red} \mathbf{100.0 \pm 0.0}$ & $\color{red} \mathbf{100.0 \pm 0.1}$ & $81.6 \pm 5.3$ & - & - \\

    \hline
    $\noise{\texttt{LS}}$ & $66.7 \pm 0.0$ & $66.7 \pm 0.0$ & $76.3 \pm 2.3$ & $66.8 \pm 1.0$ & $66.6 \pm 2.3$ \\ 
$\noise{\texttt{PI}}$ & $33.3 \pm 0.0$ & $65.0 \pm 1.3$ & $\color{red} \mathbf{77.4 \pm 4.4}$ & - & - \\ 

    \hline
        UCR Data                & $\class_{\gudhi}^{\rm XGB}$ & $\class_{\dtm{\gudhi}}^{\rm XGB}$ & $\class_{\RipsNet}^{\rm XGB}$ & $\knn_{\rm D}$ & $\knn_{\rm E}$ \\                           
    \hline
    $\texttt{P}$ & $70.5 \pm 0.0$ & $56.2 \pm 0.0$ & $\color{red} \mathbf{88.4 \pm 4.1}$ & $82.9 \pm 0.0$ & $78.1 \pm 0.0$ \\ 
$\noise{\texttt{P}}$ & $22.5 \pm 2.6$ & $\color{red} 53.9 \pm 2.5$ & $43.0 \pm 7.9$ & $\mathbf{82.9 \pm 0.0}$ & $78.1 \pm 0.6$ \\ 
$\texttt{SAIBORS2}$ & $63.6 \pm 0.0$ & $66.2 \pm 0.0$ & $\color{red} \mathbf{80.2 \pm 5.2}$ & $73.8 \pm 0.0$ & $72.4 \pm 0.0$ \\ 
$\noise{\texttt{SAIBORS2}}$ & $56.8 \pm 0.8$ & $60.0 \pm 1.2$ & $\color{red} \mathbf{75.6 \pm 6.6}$ & $73.7 \pm 0.9$ & $72.4 \pm 0.4$ \\ 
$\texttt{ECG5000}$ & $84.2 \pm 0.0$ & $86.2 \pm 0.0$ & $\color{red} 90.2 \pm 0.2$ & $\mathbf{93.0 \pm 0.0}$ & $92.8 \pm 0.0$ \\ 
$\noise{\texttt{ECG5000}}$ & $68.9 \pm 0.8$ & $71.6 \pm 1.0$ & $\color{red} 75.8 \pm 4.7$ & $\mathbf{93.1 \pm 0.3}$ & $92.8 \pm 0.1$ \\ 
$\texttt{UMD}$ & $55.6 \pm 0.0$ & $54.2 \pm 0.0$ & $\color{red} \mathbf{71.1 \pm 6.5}$ & $68.8 \pm 0.0$ & $61.1 \pm 0.0$ \\ 
$\noise{\texttt{UMD}}$ & $51.8 \pm 1.9$ & $48.9 \pm 1.6$ & $\color{red} \mathbf{69.2 \pm 6.4}$ & $68.3 \pm 1.7$ & $61.1 \pm 0.4$ \\ 
$\texttt{GPOVY}$ & $\color{red} 98.4 \pm 0.0$ & $97.8 \pm 0.0$ & $90.4 \pm 19.0$ & $\mathbf{100.0 \pm 0.0}$ & $100.0 \pm 0.0$ \\ 
$\noise{\texttt{GPOVY}}$ & $54.8 \pm 0.7$ & $54.3 \pm 0.6$ & $\color{red} 82.4 \pm 20.7$ & $\mathbf{100.0 \pm 0.0}$ & $100.0 \pm 0.0$ \\
    \hline
        $\lambda$ ($\%$)                          
        & $\class_{\gudhi}^{\scriptscriptstyle{\mathrm{NN}}}$ & $\class_{\dtm{\gudhi}}^{\scriptscriptstyle{\mathrm{NN}}}$     & $\class_{\RipsNet}^{\scriptscriptstyle{\mathrm{NN}}}$  & \texttt{pointnet} &  \\ 
    \hline
        $0$         & $30.4 \pm 4.0$    & $30.9 \pm 2.0$              & $\color{red} 53.9 \pm 2.4 $           & $\mathbf{81.6\pm 1.1}$ & \\
        $2$      & $30.3 \pm 3.2$    & $31.0 \pm 2.7$              & $\color{red} 53.2 \pm 2.5 $           & $\mathbf{74.5\pm 1.6}$  & \\
        $5$      & $29.9 \pm 4.0$    & $31.0 \pm 2.7$              & $\color{red} 55.1 \pm 3.3$            & $\mathbf{63.4\pm 1.6}$  & \\
        $10$       & $25.2 \pm 3.2$    & $29.5 \pm 3.1$              & $\color{red} \mathbf{51.0 \pm 2.1}$   & $50.6\pm 1.5$  & \\
        $15$      & $22.9 \pm 4.6$    & $25.7 \pm 3.1$              & $\color{red} \mathbf{46.9 \pm 3.0}$   & $44.9\pm 1.7$  & \\
        $25$      & $14.4 \pm 4.0$    & $18.1 \pm 2.6$              & $\color{red} \mathbf{42.6 \pm 2.5}$   & $11.0\pm 0.2$  & \\
        $50$       & $14.0 \pm 3.4$    & $13.1 \pm 1.9$              & $\color{red} \mathbf{31.6 \pm 3.3} $  & $10.9\pm 0.0$  & \\
        \hline
    \end{tabular}
    \caption{Accuracy scores of
    classifiers trained on $\gudhi$, $\dtm{\gudhi}$ and $\RipsNet$ PVs generated from several data sets. 
    The highest accuracy of the three topology-based classifiers (middle) 
    is highlighted in red, and the highest accuracy over all models
    in bold font.}
    \label{tab:results}
\end{table*}


\begin{table}[]
    \centering
    \begin{tabular}{|c|ccc|}
    \hline
        Data                & $\gudhi$ (s) & $\dtm{\gudhi}$ (s) & $\RipsNet$ (s) \\                           
    \hline
    $\texttt{LS}$ & $56.3$ $\pm$ $1.5$ & $155.9$ $\pm$ $8.1$ & $\mathbf{0.3 \pm 0.0}$ \\ 
$\texttt{PI}$ & $69.5$ $\pm$ $3.1$ & $173.7$ $\pm$ $13.3$ & $\mathbf{0.4 \pm 0.0}$ \\ 

    \hline
    $\texttt{P}$ & $5.3 \pm 1.4$ & $44.7 \pm 6.6$ & $ \mathbf{0.2 \pm 0.0}$ \\ 
    $\texttt{UMD}$ & $8.0 \pm 1.4$ & $55.7 \pm 3.6$ & $ \mathbf{0.2 \pm 0.0}$ \\ 
    \hline
    $\lambda= 2\%$ & $118.4\pm4.7$     & $178.5\pm8.1$     & $\mathbf{0.2 \pm 0.0}$   \\
    $\lambda= 5\%$ & $117.8\pm4.5$     & $180.0\pm9.2$     & $\mathbf{0.2 \pm 0.0}$   \\
    \hline
    \end{tabular}
    \caption{Running times for $\gudhi$, $\dtm{\gudhi}$ and $\RipsNet$.
    The bottom two rows refer to the 3D-shape experiments.}
    \label{tab:synth_time}
\end{table}

As one can see from the figure and the tables, $\RipsNet_{\rm synth}$ manages to learn features that look like reasonable PD vectorizations, and that perform reasonably well on clean data. However, features generated by $\RipsNet_{\rm synth}$ are much more robust;
even though $\class_{\gudhi}^{\rm XGB}$ and $\class_{\dtm{\gudhi}}^{\rm XGB}$ see their accuracies largely decrease when noise is added, $\class_{\RipsNet_{\rm synth}}^{\rm XGB}$ accuracy only decreases slightly. 
Note that the decrease of accuracy is more moderate for $\class_{\dtm{\gudhi}}^{\rm XGB}$ since $\dtm{\gudhi}$ is designed to be more robust to outliers.
Running times are much more favorable for $\RipsNet_{\rm synth}$, with an improvement of 2 (resp. 3)  orders of magnitude over $\gudhi$ (resp. $\dtm{\gudhi}$) both for persistence images and landscapes.

\subsection{Time series data}

\textbf{Data set.} We apply our experimental setup on several data sets from the $\ucr$ archive, which contains data sets of time series separated into train and test sets. We first converted the time series into point clouds in $\R^3$ using time-delay embedding with skip $1$ and delay $1$ with $\gudhi$, and used the first half of the train set for training
RipsNet architectures $\RipsNet$, while the second half was used for training XGBoost classifiers. The amount of corrupted points was set up as $2\%$, i.e.~$\lambda = 0.02$.

The hyperparameters were estimated exactly like for the synthetic data (see Section~\ref{sec:synth}), except that the final RipsNet architecture $\RipsNet_{\rm ucr}$ was found with $10$-fold cross-validation across several models similar to the one used in Section~\ref{sec:synth} and that persistence diagrams were computed in dimensions $0$ and $1$.
They were optimized with Adam optimizer with $\epsilon=5\cdot 10^{-4}$ and early stopping after $200$ epochs with less than $10^{-5}$ improvement.
We also focused on the first five persistence landscapes of resolution $50$ only. The baseline is made of two default $k$-nearest neighbors classifiers from $\scikit$, trained directly on the time series: one (named $\knn_{\rm E}$) computed with Euclidean distance, and one (named $\knn_{\rm D}$) computed with dynamic time warping. 

\textbf{Results.} Accuracies and running times (averaged over 10 runs) are given in Tables~\ref{tab:results} and~\ref{tab:synth_time}, and a more complete set of results (as well as the full data set names) can be found in Appendix~\ref{app:add_results}.
%
%
%
As in the synthetic experiment, $\RipsNet_{\rm ucr}$ learns valuable topological features, which often perform better than $\gudhi$
and $\dtm{\gudhi}$ on clean data, is more robust than $\gudhi$ and $\dtm{\gudhi}$ on noisy data, and is much faster to compute. 
The fact that $\RipsNet_{\rm ucr}$ often achieves better scores than $\gudhi$ and $\dtm{\gudhi}$ on clean data comes from the fact that the features learned by $\RipsNet_{\rm ucr}$ are more robust and less complex; while $\gudhi$ encodes all the topological patterns in the data, some of which are potentially due to noise, $\RipsNet_{\rm ucr}$ only retains the most salient patterns when minimizing the MSE during training.
Note however that when the training data set is too small for $\RipsNet$ to train properly, robustness can be harder to reach, as is the case for the \texttt{Plane} data set. 

\subsection{3D-shape data}
\textbf{Data set.}
In addition to investigating time series data, we ran experiments on Princeton's $\modelnet{10}$ data set, comprised of 3D-shape data of objects in 10 classes.
In order to obtain point clouds in $\R^3$, we sample $1024$ points on the surfaces of the 3D objects.
They are subsequently centered and normalized to be contained in the unit sphere.
We have $2393/598$ and $406/229$ training and test samples at our disposal for the training of RipsNet, 
and for the training of neural net classifiers, respectively.
The architecture of these classifiers ($\mathrm{NN}$) is very simple, consisting of only two consecutive fully connected layers of $100$ and $50$ neurons and an output layer.
In addition to the neural net classifiers, we also train XGBoost classifiers, the results of which, as well as additional results of the neural net classifiers and running times, are reported in Tables~\ref{tab:modelnet10_acc_NN_full}, \ref{tab:modelnet10_acc_xgb} and \cref{tab:modelnet10_time} in Appendix~\ref{app:add_results}.

For the sake of simplicity, we focus on persistence images of resolution $25\times 25$ with weight function $(y-x)^2$ only, and consider the combination of persistence diagrams of dimension $0$ and $1$.
The vectorization parameters were estimated as in Section~\ref{sec:synth} (due to computational cost, only on a random subset of all PDs).
The final RipsNet architecture, using $\op=$ mean, was found via a $3$-fold cross-validation over several models, and again optimized with Adam optimizer. 
As a baseline, we employ the \texttt{pointnet} model \cite{pointnet}\footnote{For an implementation of \texttt{pointnet} see: \url{https://github.com/keras-team/keras-io/blob/master/examples/vision/pointnet.py}}.
To showcase the robustness of $\RipsNet$, we introduce noise fractions $\lambda$ in $\{0.02, 0.05, 0.1, 0.25, 0.5\}$.

\textbf{Results.}
The accuracies of the $\mathrm{NN}$ classifier are compiled in \cref{tab:results} and some running times can be found in \cref{tab:synth_time}.
Due to class imbalances, the accuracy of the best possible \emph{constant} classifier is $22.2\%$.
As the sampling of the point clouds, as well as the addition of noise, are random, we repeat this process $10$ times in total.
Subsequently, we train the classifiers on each of these data sets, without retraining RipsNet, and report the mean and standard deviation.
The vectorization running time of $\gudhi$ is clearly outperformed by $\RipsNet$ by three orders of magnitude.
The accuracy of $\class_{\RipsNet}^{\scriptscriptstyle{\mathrm{NN}}}$ substantially surpasses those of $\class_{\gudhi}^{\scriptscriptstyle{\mathrm{NN}}}$ and $\class_{\dtm{\gudhi}}^{\scriptscriptstyle{\mathrm{NN}}}$ for all values of $\lambda$ and remains much more robust for high levels of noise.
For $\lambda \geq 0.1$, $\class_{\RipsNet}^{\scriptscriptstyle{\mathrm{NN}}}$ surpasses the \texttt{pointnet} baseline, whose accuracy decreases sharply for $\lambda \geq 0.25$, at which point $\class_{\RipsNet}^{\scriptscriptstyle{\mathrm{NN}}}$ substantially outperforms \texttt{pointnet}.

\section{Conclusion}
The computational complexity of the exact computation of persistence diagrams and their sensitivity to outliers and noise limit their applicability.
The vectorization of topological features of data sets is of central importance for their practical use in machine learning applications.
To address these limitations, we propose RipsNet, a Deep Sets-like architecture, that learns to estimate persistence vectorizations of point cloud data.
We prove theoretical results, showing that RipsNet is more robust to outliers and noise than exact vectorization.
Moreover, we substantiate our theoretical findings by numerical experiments on a synthetic data set, as well as on real-world time series data and 3D-shape data.
Thereby we show the robustness advantages and significant improvement in running times of RipsNet.
As inherent to machine learning frameworks, RipsNet is dependent on a successful training stage and hyperparameter tuning.
We envision that our work will find its applications in settings where pertinent topological features are obscured by the presence of noise or outliers in exact computations.
RipsNet is better suited to retain representations of such salient patterns in its estimations, due to its demonstrated robustness properties.
Another area of application is in situations where, due to computational cost, vectorizations of topological features are infeasible to compute exactly.

\paragraph*{Acknowledgments.}

The authors are grateful to the OPAL infrastructure from Université Côte d'Azur for providing resources and support.

\clearpage
\bibliography{arxiv}
\bibliographystyle{alpha}

\clearpage
\appendix
\onecolumn

\section{Fundamentals of Topological Data Analysis}\label{app:filtrations}

In this section, we briefly explain some fundamentals of topological data analysis.

\subsection{Simplicial complexes and homology groups}
Let us begin by introducing the concepts of simplicial complexes and homology groups.

\begin{definition}\label{simplicialcomplex}
    Let $V$ be a finite set. 
    A subset $K$ of the power set $P(V)$ is said to be a (finite) \emph{simplicial complex} with vertex set $V$ if it satisfies the following conditions.
    \begin{enumerate}
        \item[(1)] $\emptyset \not\in K$;
        \item[(2)] for any $v \in V$, $\{v\} \in K$;\label{def:simplex-vertices}
        \item[(3)] if $\sigma \in K$ and $\emptyset \neq \tau \subset \sigma$, then $\tau \in K$.
    \end{enumerate}
    An element $\sigma$ of $K$ with the cardinality $\# \sigma=k+1$ is called a $k$-simplex.
\end{definition}

Now we introduce homology to extract topology information of simplical complexes.
For $\sigma = \{v_{0}, \dots, v_{k}\}$, we consider orderings with respect to its vertices. 
Two orderings of $\sigma$ are said to be \emph{equivalent} if one ordering can be obtained from the other by an even permutation.
In this way, orderings consist of two equivalence classes, each of which is called an \emph{orientation} of $\sigma$. 
A simplex equipped with an orientation is said to be oriented, and a simplicial complex whose simplices are all oriented is called an oriented simplicial complex.
The equivalent class of the ordering $(v_{i_0}, \dots, v_{i_k})$ is denoted by $\langle v_{i_0},\dots, v_{i_k} \rangle$.
We use the convention that 
\[
    \langle v_{s(0)}, \dots, v_{s(k)} \rangle = \operatorname{sgn}(s) \langle v_{0}, \dots, v_{k} \rangle
\]
for any permutation $s$ of $\{0,\dots,k\}$, where $\operatorname{sgn}(s)$ denotes the signature of $s$.
For an oriented simplicial complex $K$, let $C_{i}(K)$ be the free abelian group consists of equivalence classes of oriented $i$-simplices of $K$, which is called the $i$-th $\emph{chain group}$ of $K$.
We now introduce a homomorphism, which is called the \emph{boundary operator}.

\begin{definition}
Let $K$ be an oriented simplicial complex and $i$ be a positive integer.
For a $i$-simplex $\sigma \in C_{i}(K)$, one defines the boundary operator $\partial_{i}\colon C_{i}(K)\to C_{i-1}(K)$ by 
\begin{equation*}
    \partial_{i}(\sigma) = \sum_{j=0}^{k}(-1)^{j} \langle v_{0} \dots v_{j-1}v_{j+1}\dots v_{n}\rangle.
\end{equation*}
Then one linearly extends the operator for the elements of $C_{i}(K)$.
One also sets $\partial_{0}=0$.
\end{definition}

Then we can see that $\partial_{i}\circ\partial_{i+1}=0$ for any non-negative integer.
This implies $\operatorname{Im} (\partial_{i+1})\subset \operatorname{Ker}(\partial_{i})$.

\begin{definition}
    For an oriented simplicial complex $K$, one defines  
    \[
        H_{i}(K)\coloneqq \operatorname{Ker}(\partial_{i})/\operatorname{Im}(\partial_{i+1})
    \]
    and calls it the $i$-th \emph{homology group} of $K$.
\end{definition}

We note that a simplicial complex defined in Definition~\ref{simplicialcomplex} is sometimes called an \emph{abstract simplicial complex}.
As explained in Definition \ref{simplicialcomplex}, a simplicial complex consists of a finite set and its power set.
We remark that the finite set $V$ does not need to be a subset of Euclidean space.
On the other hand, we can interpret the simplicial complex from a geometric perspective.
Let $K$ be a simplicial complex and $V$ the vertex set of $K$.
Set $N \coloneqq \# V$ and consider the $N$-dimensional vector space $\R^N$ with standard basis $e_i \coloneqq (\underbrace{0,  \dots, 0}_{i-1}, 1, 0, \dots, 0)$. 
With each simplex $\sigma = \{v_{i_0}, \dots, v_{i_k}\}$ of $K$, we associate a $k$-simplex 
\[
    |\sigma| \coloneqq \{\lambda_{i_0}e_{i_0}+\dots+\lambda_{i_k}e_{i_k}\mid\lambda_{i_0}+\dots+\lambda_{i_k}=1, \lambda_{i}\geq0\}.
\]
We define the \emph{geometric realization} of $K$ by $|K| \coloneqq \bigcup_{\sigma \in K} |\sigma| \subset \R^N$ with the subspace topology.
A fundamental theorem is that the simplicial homology $K$ is isomorphic to the singular homology of its geometric realization $|K|$: $H_n(K) \simeq H_n(|K|)$ for any $n \in \Z$.
Hence, for a topological space $\XX$, if we find a simplical complex $K$ such that its geometric realization $|K|$ is homotopy equivalent to $\XX$, we can compute the singular homology $H_n(\XX)$ by the combinatorial object $K$. 
Below, we will construct such a complex for a finite union of closed balls. 

\subsection{Filtrations and persistent homology}

We introduce the notion of a filtration of a simplicial complex to consider the evolution of the topology of the simplicial complex. 

\begin{definition}\label{def:filtration}
    Let $K$ be a simplicial complex.
    A family of subcomplexes $(K_\alpha)_{\alpha \in \R}$ of $K$ is said to be a \emph{filtration} of $K$ if it satisfies 
    \begin{enumerate}
        \item[(1)] $K_\alpha \subset K_{\alpha'}$ for $\alpha \le \alpha'$;
        \item[(2)] $\bigcup_{\alpha \in \R} K_\alpha=K$. 
    \end{enumerate}
\end{definition}

\textbf{\v{C}ech filtration.} Let $X$ be a finite point set in $\XX=\R^d$ and $\alpha \in \R$. 
To $X$, one can associate a function $d_X \colon \R^d \to \R, v \mapsto d(v,X):=\min_{x \in X} \|v-x\|$. 
The sublevel filtration induced by $d_X$ is frequently used to investigate the topology of the point set $X$.
Indeed, for any $\alpha \ge 0$ the sublevel set $\XX_\alpha=\{ x \in \R^d : d_X(x) \le \alpha\}$ is equal to the union of $d$-dimensional closed balls of radius $\alpha$ centered at points in $X$: $\bigcup_{x \in X} \overline{B}(x;\alpha)$.
One can compute the homology of the sublevel set in the following combinatorial way, by constructing a simplicial complex whose geometric realization is homotopy equivalent to the sublevel set. 
For $\alpha \ge 0$, we define a simplicial complex $\CC(X;\alpha)$ by 
\[
    \{x_0,\dots,x_k\} \in \CC(X;\alpha) \vcentcolon\iff \bigcap_{i=0}^{k} \overline{B}(x_i;\alpha) \neq \emptyset.
\]
In other words, $\CC(X;\alpha)$ is the nerve of the family of closed sets $\{ \overline{B}(x;\alpha) \}_{x \in X}$.
Since each $\overline{B}(x;\alpha)$ is a convex closed set of $\R^d$, the geometric realization of $\CC(X;\alpha)$ is homotopy equivalent to $\bigcup_{x \in X} \overline{B}(x;\alpha)$ by the nerve theorem. 
The family $(\CC(X;\alpha))_\alpha$ forms a filtration in the sense of \cref{def:filtration} for $\alpha \geq 0$.
If $\alpha$ is negative, we regard $\CC(X;\alpha)$ as the empty set.
We call this filtration the \emph{\v{C}ech filtration}.

In the above construction of filtrations, the radii of balls increase uniformly.
We can give filtrations in another way, that is, we make radii increase non-uniformly.
Such a filtration is called a \emph{weighted filtration}.
Let $X \subset \mathbb{R}^{d}$ be a finite point set, $f \colon \mathbb{R}^{d} \to \mathbb{R}_{\geq 0}$ a continuous function, and $p \in [1,\infty]$.
For $p <\infty$, we define a function 
$r_f \colon X \times \mathbb{R}_{\geq 0} \to \mathbb{R} \cup \{-\infty\}$ by 
\[
r_{f}(x,t) = 
\begin{cases}
    -\infty & \text{if }t < f(x),\\
    (t^{p}-f(x)^{p})^{\frac{1}{p}} & \text{otherwise.}
\end{cases}
\]
When $p=\infty$, we also define a function $r_f \colon X \times \mathbb{R}_{\geq 0} \to \mathbb{R} \cup \{-\infty\}$ by 
\[
r_{f}(x,t) =
\begin{cases}
    -\infty & \text{if }t < f(x),\\
t & \text{otherwise.}
\end{cases}
\]
We replace the radius of each closed ball $\overline{B}(x;r)$ by $r_{f}(x,t)$.
By modifying the definition of $\CC(X;\alpha)$, we define a simplicial complex $\CC_{f}(X;t)$ by 
\[
    \{x_0,\dots,x_k\} \in \CC_{f}(X;t) \vcentcolon \iff \bigcap_{i=0}^{k} \overline{B}(x_i;r_{f}(x_{i},t)) \neq \emptyset.
\]
Then we have a filtration $\{\CC_{f}(X;t)\}_t$, which is called the \emph{weighted \v{C}ech filtration}.

\textbf{Rips filtration.} The \v{C}ech complex $\CC(X;\alpha)$ exactly computes the homology of the union of closed balls $\bigcup_{x \in X} \overline{B}(x;\alpha)$, but it is computationally expensive in practice. 
Now we introduce another simplicial complex that is less expensive than the \v{C}ech complex.

Let $X$ be a finite point set in $\R^d$.
For any $\alpha \geq 0$, one can define a simplicial complex $\mathcal{R}(X;\alpha)$ whose vertex set is $X$ by 
\begin{align*}
    \{x_0,\dots,x_k\} \in \mathcal{R}(X;\alpha) 
    & \vcentcolon\iff
    \overline{B}(x_i;\alpha) \cap \overline{B}(x_j;\alpha) \neq \emptyset \ \text{for any $i,j \in \{0,\dots,k\}$}\\
    & \iff \|x_i-x_k\| \le 2\alpha \ \text{for any $i,j \in \{0,\dots,k\}$}.
\end{align*}
Otherwise, we regard $\mathcal{R}(X;\alpha)$ as the empty set.
The family $(\mathcal{R}(X;\alpha))_{\alpha}$ forms a filtration, which we call the \emph{Rips filtration}.

Remark that we can also construct the \emph{weighted Rips filtration} similarly to the weighted \v{C}ech filtration.

\textbf{Alpha filtration.} Let $X$ be a finite subset of $\R^d$ and assume that $X$ is in a general position.
For $x \in X$, set 
\[
    V_x \coloneqq \{ y \in \R^d \mid d(x,y) \le d(x',y), \text{for any $x' \in X \setminus \{x\}$} \},
\]
which we call the Voronoi cell for $x$.
With this notation, we set
\[
    W(x;\alpha) \coloneqq B(x;\alpha) \cap V_x, \qquad F(X;\alpha) \coloneqq \{ W(x;\alpha) \}_{x \in X}
\]
and define $\mathrm{Alpha}(X;\alpha)$ to be the nerve of $F(X;\alpha)$.
Then by the nerve theorem, we find that the geometric realization of $\mathrm{Alpha}(X;\alpha)$ is homotopy equivalent to $\bigcup_{x \in X} \overline{B}(x;\alpha)$.
The family $(\mathrm{Alpha}(X;\alpha))_{\alpha}$ forms a filtration and is called the alpha filtration.
Similarly, one can construct a weighted version of the alpha filtration, called the \emph{weighted alpha filtration}.

\textbf{DTM-filtrations.} To get robustness to noise and outliers, \cite{Anai2019} introduces DTM-based filtrations. 
Let $\mu$ be a probability measure on $\mathbb{R}^d$ and $m$ a parameter in $[0,1)$.
We define a function $\delta_{\mu,m} \colon \mathbb{R}^{d} \to \mathbb{R}$ by $\delta_{\mu,m}(x) = \text{inf}\{r \geq 0, \mu(\overline{B}(x;r)) > m\}$.
\begin{definition}
    The distance-to-measure function (DTM for short) $\mu$ is the function $d_{\mu,m} \colon \mathbb{R}^{d} \to \mathbb{R}$ defined by $$d_{\mu,m}(x) = \sqrt{\frac{1}{m}\int_{0}^{m}\delta_{\mu,t}^{2}(x)dt}.$$
\end{definition}

Let $X$ be a finite point cloud with $n$ points and set $\mu_X$ to be a empirical measure associated with $X$: $\mu_X=\frac{1}{n} \sum_{x \in X} \delta_{x}$, where $\delta_x$ is the Dirac measure located at $x$.
For fixed $p \in [0,+\infty]$ and $m \in [0,1)$, the family of simplicial complexes $\CC_{d_{\mu_X, m}}(X;t)$ defines a filtration, which we call the \emph{DTM-based filtration}.
By replacing the \v{C}ech filtration with the alpha filtration, we can also define the AlphaDTM-based filtration.
These filtrations are shown to be robust to noise and outliers.

\textbf{Persistent homology.}
Given a filtration $(K_{\alpha})_{\alpha \in \mathbb{R}}$ of an $n$-dimensional simplicial complex $K$, we have inclusion maps $\iota_{\alpha}^{\alpha'} \colon K_{\alpha} \hookrightarrow K_{\alpha'}$
for any $\alpha \leq \alpha'$.
Such inclusion maps induce
homomorphisms $(\iota_{\alpha}^{\alpha'})_{\ast} \colon H_{i}(K_{\alpha}) \to H_{i}(K_{\alpha'})$.
Then we have a family of homomorphisms $\cdots \to H_{i}(K_{\alpha}) \to H_{i}(K_{\alpha'})\to \cdots$.
The resulting family is called $i$-th \emph{persistence module}, and is known to decompose into simpler interval modules, that can be represented as a persistence diagram.

\section{Additional Experimental Results and Larger Figures} \label{app:add_results}

\begin{table*}[]
    \centering
    \begin{tabular}{|c|ccc|cc|}
    \hline
        Data                & $\class_{\gudhi}^{\rm XGB}$ & $\class_{\dtm{\gudhi}}^{\rm XGB}$ & $\class_{\RipsNet}^{\rm XGB}$ & $\knn_{\rm D}$ & $\knn_{\rm E}$ \\                           
    \hline
    $\texttt{CC}$ & $53.4 \pm 0.0$ & $52.0 \pm 0.0$ & $\color{red} \mathbf{55.8 \pm 1.3}$ & $55.2 \pm 0.0$ & $53.2 \pm 0.0$ \\ 
$\noise{\texttt{CC}}$ & $52.9 \pm 0.6$ & $\color{red} 53.7 \pm 0.2$ & $45.3 \pm 2.0$ & $\mathbf{54.5 \pm 0.7}$ & $51.7 \pm 0.5$ \\ 
$\texttt{PPTW}$ & $76.1 \pm 0.0$ & $71.2 \pm 0.0$ & $\color{red} 76.5 \pm 1.1$ & $\mathbf{78.5 \pm 0.0}$ & $76.1 \pm 0.0$ \\ 
$\noise{\texttt{PPTW}}$ & $\color{red} 70.7 \pm 2.3$ & $68.0 \pm 2.1$ & $62.0 \pm 7.6$ & $\mathbf{78.2 \pm 1.0}$ & $73.1 \pm 1.6$ \\ 
$\texttt{P}$ & $70.5 \pm 0.0$ & $56.2 \pm 0.0$ & $\color{red} \mathbf{88.4 \pm 4.1}$ & $82.9 \pm 0.0$ & $78.1 \pm 0.0$ \\ 
$\noise{\texttt{P}}$ & $22.5 \pm 2.6$ & $\color{red} 53.9 \pm 2.5$ & $43.0 \pm 7.9$ & $\mathbf{82.9 \pm 0.0}$ & $78.1 \pm 0.6$ \\ 
$\texttt{GP}$ & $80.7 \pm 0.0$ & $\color{red} \mathbf{84.0 \pm 0.0}$ & $75.7 \pm 6.6$ & $66.7 \pm 0.0$ & $72.7 \pm 0.0$ \\ 
$\noise{\texttt{GP}}$ & $50.0 \pm 0.7$ & $50.5 \pm 0.5$ & $\color{red} 68.1 \pm 5.3$ & $66.8 \pm 0.3$ & $\mathbf{72.8 \pm 0.5}$ \\ 
$\texttt{POC}$ & $68.2 \pm 0.0$ & $64.2 \pm 0.0$ & $\color{red} 71.6 \pm 2.8$ & $72.6 \pm 0.0$ & $\mathbf{74.6 \pm 0.0}$ \\ 
$\noise{\texttt{POC}}$ & $\color{red} 58.6 \pm 0.0$ & $57.4 \pm 0.0$ & $53.2 \pm 0.0$ & $71.8 \pm 0.0$ & $\mathbf{74.2 \pm 0.0}$ \\ 
$\texttt{SAIBORS2}$ & $63.6 \pm 0.0$ & $66.2 \pm 0.0$ & $\color{red} \mathbf{80.2 \pm 5.2}$ & $73.8 \pm 0.0$ & $72.4 \pm 0.0$ \\ 
$\noise{\texttt{SAIBORS2}}$ & $56.8 \pm 0.8$ & $60.0 \pm 1.2$ & $\color{red} \mathbf{75.6 \pm 6.6}$ & $73.7 \pm 0.9$ & $72.4 \pm 0.4$ \\ 
$\texttt{PPOAG}$ & $78.5 \pm 0.0$ & $79.5 \pm 0.0$ & $\color{red} 81.1 \pm 2.8$ & $\mathbf{82.9 \pm 0.0}$ & $82.9 \pm 0.0$ \\ 
$\noise{\texttt{PPOAG}}$ & $\color{red} 74.6 \pm 1.9$ & $72.7 \pm 1.5$ & $73.9 \pm 3.6$ & $82.0 \pm 0.4$ & $\mathbf{83.1 \pm 1.3}$ \\ 
$\texttt{ECG5000}$ & $84.2 \pm 0.0$ & $86.2 \pm 0.0$ & $\color{red} 90.2 \pm 0.2$ & $\mathbf{93.0 \pm 0.0}$ & $92.8 \pm 0.0$ \\ 
$\noise{\texttt{ECG5000}}$ & $68.9 \pm 0.8$ & $71.6 \pm 1.0$ & $\color{red} 75.8 \pm 4.7$ & $\mathbf{93.1 \pm 0.3}$ & $92.8 \pm 0.1$ \\ 
$\texttt{ECG200}$ & $\color{red} 77.0 \pm 0.0$ & $70.0 \pm 0.0$ & $76.2 \pm 1.6$ & $78.0 \pm 0.0$ & $\mathbf{85.0 \pm 0.0}$ \\ 
$\noise{\texttt{ECG200}}$ & $\color{red} 73.0 \pm 2.9$ & $70.2 \pm 4.0$ & $72.8 \pm 2.7$ & $78.4 \pm 0.5$ & $\mathbf{85.0 \pm 0.6}$ \\ 
$\texttt{MI}$ & $47.2 \pm 0.0$ & $46.8 \pm 0.0$ & $\color{red} 56.4 \pm 2.2$ & $\mathbf{63.4 \pm 0.0}$ & $52.6 \pm 0.0$ \\ 
$\noise{\texttt{MI}}$ & $35.1 \pm 2.2$ & $36.6 \pm 2.0$ & $\color{red} 44.3 \pm 3.4$ & $\mathbf{63.6 \pm 0.3}$ & $53.6 \pm 0.5$ \\ 
$\texttt{PC}$ & $69.4 \pm 0.0$ & $74.4 \pm 0.0$ & $\color{red} 87.6 \pm 6.0$ & $82.8 \pm 0.0$ & $\mathbf{97.8 \pm 0.0}$ \\ 
$\noise{\texttt{PC}}$ & $70.0 \pm 2.7$ & $72.7 \pm 2.4$ & $\color{red} 84.8 \pm 5.3$ & $84.2 \pm 1.4$ & $\mathbf{97.9 \pm 0.2}$ \\ 
$\texttt{DPOC}$ & $68.5 \pm 0.0$ & $69.2 \pm 0.0$ & $\color{red} \mathbf{74.0 \pm 2.0}$ & $73.9 \pm 0.0$ & $68.8 \pm 0.0$ \\ 
$\noise{\texttt{DPOC}}$ & $57.8 \pm 2.1$ & $58.7 \pm 1.7$ & $\color{red} 61.5 \pm 3.6$ & $\mathbf{73.9 \pm 0.5}$ & $72.0 \pm 0.9$ \\ 
$\texttt{IPD}$ & $70.4 \pm 0.0$ & $70.6 \pm 0.0$ & $\color{red} 79.0 \pm 0.8$ & $87.6 \pm 0.0$ & $\mathbf{96.4 \pm 0.0}$ \\ 
$\noise{\texttt{IPD}}$ & $70.4 \pm 0.0$ & $70.6 \pm 0.0$ & $\color{red} 79.0 \pm 0.8$ & $87.6 \pm 0.0$ & $\mathbf{96.4 \pm 0.0}$ \\ 
$\texttt{MPOAG}$ & $54.5 \pm 0.0$ & $50.0 \pm 0.0$ & $\color{red} 55.6 \pm 1.9$ & $\mathbf{56.5 \pm 0.0}$ & $52.0 \pm 0.0$ \\ 
$\noise{\texttt{MPOAG}}$ & $35.3 \pm 4.6$ & $33.0 \pm 2.8$ & $\color{red} 49.7 \pm 7.4$ & $\mathbf{56.0 \pm 0.5}$ & $53.2 \pm 2.5$ \\ 
$\texttt{SAIBORS1}$ & $53.0 \pm 0.0$ & $49.8 \pm 0.0$ & $\color{red} \mathbf{69.5 \pm 3.1}$ & $50.2 \pm 0.0$ & $45.0 \pm 0.0$ \\ 
$\noise{\texttt{SAIBORS1}}$ & $53.1 \pm 0.8$ & $49.5 \pm 0.7$ & $\color{red} \mathbf{67.9 \pm 5.8}$ & $50.2 \pm 0.2$ & $44.9 \pm 0.2$ \\ 
$\texttt{UMD}$ & $55.6 \pm 0.0$ & $54.2 \pm 0.0$ & $\color{red} \mathbf{71.1 \pm 6.5}$ & $68.8 \pm 0.0$ & $61.1 \pm 0.0$ \\ 
$\noise{\texttt{UMD}}$ & $51.8 \pm 1.9$ & $48.9 \pm 1.6$ & $\color{red} \mathbf{69.2 \pm 6.4}$ & $68.3 \pm 1.7$ & $61.1 \pm 0.4$ \\ 
$\texttt{TLECG}$ & $67.6 \pm 0.0$ & $71.8 \pm 0.0$ & $\color{red} \mathbf{78.6 \pm 11.0}$ & $75.8 \pm 0.0$ & $54.0 \pm 0.0$ \\ 
$\noise{\texttt{TLECG}}$ & $66.6 \pm 1.0$ & $\color{red} 69.6 \pm 0.9$ & $69.5 \pm 6.6$ & $\mathbf{74.7 \pm 0.4}$ & $54.2 \pm 0.3$ \\ 
$\texttt{MPOC}$ & $68.7 \pm 0.0$ & $66.7 \pm 0.0$ & $\color{red} 73.2 \pm 2.0$ & $74.2 \pm 0.0$ & $\mathbf{77.7 \pm 0.0}$ \\ 
$\noise{\texttt{MPOC}}$ & $64.7 \pm 1.4$ & $64.2 \pm 1.7$ & $\color{red} 65.0 \pm 2.6$ & $73.3 \pm 0.8$ & $\mathbf{76.7 \pm 0.9}$ \\ 
$\texttt{GPOVY}$ & $\color{red} 98.4 \pm 0.0$ & $97.8 \pm 0.0$ & $90.4 \pm 19.0$ & $\mathbf{100.0 \pm 0.0}$ & $100.0 \pm 0.0$ \\ 
$\noise{\texttt{GPOVY}}$ & $54.8 \pm 0.7$ & $54.3 \pm 0.6$ & $\color{red} 82.4 \pm 20.7$ & $\mathbf{100.0 \pm 0.0}$ & $100.0 \pm 0.0$ \\ 
$\texttt{MPTW}$ & $52.0 \pm 0.0$ & $\color{red} \mathbf{52.6 \pm 0.0}$ & $51.7 \pm 1.3$ & $50.6 \pm 0.0$ & $52.0 \pm 0.0$ \\ 
$\noise{\texttt{MPTW}}$ & $36.5 \pm 1.6$ & $35.2 \pm 3.5$ & $\color{red} 46.5 \pm 4.5$ & $51.0 \pm 1.3$ & $\mathbf{51.2 \pm 1.8}$ \\ 
$\texttt{CBF}$ & $\color{red} 64.4 \pm 0.0$ & $63.8 \pm 0.0$ & $61.6 \pm 10.9$ & $\mathbf{78.2 \pm 0.0}$ & $58.4 \pm 0.0$ \\ 
$\noise{\texttt{CBF}}$ & $\color{red} 63.8 \pm 1.1$ & $62.5 \pm 1.6$ & $55.0 \pm 9.3$ & $\mathbf{79.0 \pm 0.4}$ & $57.7 \pm 0.3$ \\ 

    \hline
    \end{tabular}
    \caption{Accuracy scores of XGBoost and $k$-NN classifiers on UCR data sets.}
    \label{tab:ucr_acc_long}
\end{table*}

\begin{table}[]
    \centering
    \begin{tabular}{|cc|ccc|}
    \hline
        Data                & Name & $\gudhi$ (s) & $\dtm{\gudhi}$ (s) & $\RipsNet$ (s) \\                           
    \hline
    $\texttt{CC}$ & $\texttt{ChlorineConcentration}$ & $33.6 \pm 9.3$ & $310.0 \pm 36.1$ & $\mathbf{0.2 \pm 0.0}$ \\ 
$\noise{\texttt{CC}}$ & - & $25.3 \pm 2.5$ & $298.2 \pm 43.6$ & $\mathbf{0.2 \pm 0.0}$ \\ 
$\texttt{PPTW}$ & $\texttt{ProximalPhalanxTW}$ & $5.8 \pm 1.3$ & $25.3 \pm 4.5$ & $ \mathbf{0.2 \pm 0.0}$ \\ 
$\noise{\texttt{PPTW}}$ & - & $5.9 \pm 1.3$ & $25.6 \pm 4.6$ & $ \mathbf{0.2 \pm 0.0}$ \\ 
$\texttt{P}$ & $\texttt{Plane}$ & $5.3 \pm 1.4$ & $44.7 \pm 6.6$ & $ \mathbf{0.2 \pm 0.0}$ \\ 
$\noise{\texttt{P}}$ & - & $5.3 \pm 1.5$ & $44.8 \pm 6.2$ & $ \mathbf{0.2 \pm 0.0}$ \\ 
$\texttt{GP}$ & $\texttt{GunPoint}$ & $6.6 \pm 1.5$ & $53.7 \pm 7.9$ & $ \mathbf{0.2 \pm 0.0}$ \\ 
$\noise{\texttt{GP}}$ & - & $6.6 \pm 1.5$ & $55.1 \pm 8.0$ & $ \mathbf{0.2 \pm 0.0}$ \\ 
$\texttt{POC}$ & $\texttt{PhalangesOutlineCorrect}$ & $22.0 \pm 5.1$ & $91.6 \pm 19.5$ & $ \mathbf{0.2 \pm 0.0}$ \\ 
$\noise{\texttt{POC}}$ & - & $11.9 \pm 0.0$ & $56.2 \pm 0.0$ & $ \mathbf{0.2 \pm 0.0}$ \\ 
$\texttt{SAIBORS2}$ & $\texttt{SonyAIBORobotSurface2}$ & $10.0 \pm 0.2$ & $26.3 \pm 0.6$ & $ \mathbf{0.2 \pm 0.0}$ \\ 
$\noise{\texttt{SAIBORS2}}$ & - & $10.2 \pm 0.2$ & $26.9 \pm 0.6$ & $ \mathbf{0.2 \pm 0.0}$ \\ 
$\texttt{PPOAG}$ & $\texttt{ProximalPhalanxOutlineAgeGroup}$ & $6.8 \pm 0.2$ & $28.9 \pm 0.8$ & $ \mathbf{0.2 \pm 0.0}$ \\ 
$\noise{\texttt{PPOAG}}$ & - & $6.9 \pm 0.2$ & $29.4 \pm 0.9$ & $ \mathbf{0.2 \pm 0.0}$ \\ 
$\texttt{ECG5000}$ & $\texttt{ECG5000}$ & $27.6 \pm 8.0$ & $187.1 \pm 29.2$ & $ \mathbf{0.2 \pm 0.0}$ \\ 
$\noise{\texttt{ECG5000}}$ & - & $27.2 \pm 8.6$ & $185.5 \pm 32.3$ & $ \mathbf{0.2 \pm 0.0}$ \\ 
$\texttt{ECG200}$ & $\texttt{ECG200}$ & $3.8 \pm 1.2$ & $14.9 \pm 4.3$ & $ \mathbf{0.1 \pm 0.0}$ \\ 
$\noise{\texttt{ECG200}}$ & - &  $3.8 \pm 1.1$ & $15.1 \pm 3.9$ & $\mathbf{0.2 \pm 0.0}$ \\ 
$\texttt{MI}$ & $\texttt{MedicalImages}$ & $16.8 \pm 2.5$ & $79.8 \pm 9.4$ & $ \mathbf{0.3 \pm 0.1}$ \\ 
$\noise{\texttt{MI}}$ & - & $16.5 \pm 2.4$ & $80.1 \pm 8.1$ & $ \mathbf{0.3 \pm 0.0}$ \\ 
$\texttt{PC}$ & $\texttt{PowerCons}$ & $11.2 \pm 2.6$ & $77.1 \pm 10.0$ & $ \mathbf{0.2 \pm 0.0}$ \\ 
$\noise{\texttt{PC}}$ & - & $11.3 \pm 2.8$ & $76.9 \pm 11.4$ & $ \mathbf{0.2 \pm 0.0}$ \\ 
$\texttt{DPOC}$ & $\texttt{DistalPhalanxOutlineCorrect}$ & $9.4 \pm 1.6$ & $39.2 \pm 5.3$ & $ \mathbf{0.2 \pm 0.0}$ \\ 
$\noise{\texttt{DPOC}}$ & - & $9.4 \pm 1.5$ & $39.4 \pm 5.1$ & $ \mathbf{0.2 \pm 0.0}$ \\ 
$\texttt{IPD}$ & $\texttt{ItalyPowerDemand}$ & $2.7 \pm 0.5$ & $3.9 \pm 0.6$ & $ \mathbf{0.2 \pm 0.0}$ \\ 
$\noise{\texttt{IPD}}$ & - & $2.7 \pm 0.5$ & $3.9 \pm 0.6$ & $\mathbf{0.2 \pm 0.0}$ \\ 
$\texttt{MPOAG}$ & $\texttt{MiddlePhalanxOutlineAgeGroup}$ & $5.8 \pm 1.0$ & $24.1 \pm 3.3$ & $ \mathbf{0.2 \pm 0.0}$ \\ 
$\noise{\texttt{MPOAG}}$ & - & $5.8 \pm 0.9$ & $24.1 \pm 3.2$ & $ \mathbf{0.2 \pm 0.0}$ \\ 
$\texttt{SAIBORS1}$ & $\texttt{SonyAIBORobotSurface2}$ & $10.4 \pm 1.5$ & $28.4 \pm 4.2$ & $ \mathbf{0.2 \pm 0.0}$ \\ 
$\noise{\texttt{SAIBORS1}}$ & - & $10.5 \pm 1.5$ & $28.6 \pm 4.2$ & $ \mathbf{0.2 \pm 0.0}$ \\ 
$\texttt{UMD}$ & $\texttt{UMD}$ & $8.0 \pm 1.4$ & $55.7 \pm 3.6$ & $ \mathbf{0.2 \pm 0.0}$ \\ 
$\noise{\texttt{UMD}}$ & - & $8.0 \pm 1.3$ & $55.6 \pm 2.6$ & $\mathbf{0.2 \pm 0.0}$ \\ 
$\texttt{TLECG}$ & $\texttt{TwoLeadECG}$ & $11.0 \pm 0.4$ & $40.8 \pm 1.5$ & $\mathbf{0.2 \pm 0.0}$ \\ 
$\noise{\texttt{TLECG}}$ & - & $11.1 \pm 0.3$ & $41.2 \pm 0.9$ & $ \mathbf{0.2 \pm 0.0}$ \\ 
$\texttt{MPOC}$ & $\texttt{MiddlePhalanxOutlineCorrect}$ & $9.6 \pm 2.1$ & $39.3 \pm 8.1$ & $\mathbf{0.2 \pm 0.0}$ \\ 
$\noise{\texttt{MPOC}}$ & - & $9.6 \pm 2.1$ & $39.4 \pm 7.9$ & $ \mathbf{0.2 \pm 0.0}$ \\ 
$\texttt{GPOVY}$ & $\texttt{GunPointOldVersusYoung}$ & $12.1 \pm 3.8$ & $108.1 \pm 21.5$ & $ \mathbf{0.2 \pm 0.0}$ \\ 
$\noise{\texttt{GPOVY}}$ & - & $12.3 \pm 3.8$ & $108.8 \pm 19.6$ & $ \mathbf{0.2 \pm 0.0}$ \\ 
$\texttt{MPTW}$ & $\texttt{MiddlePhalanxTW}$ & $6.3 \pm 0.2$ & $26.3 \pm 1.0$ & $ \mathbf{0.2 \pm 0.0}$ \\ 
$\noise{\texttt{MPTW}}$ & - & $6.4 \pm 0.2$ & $26.1 \pm 0.8$ & $ \mathbf{0.2 \pm 0.0}$ \\ 
$\texttt{CBF}$ & $\texttt{CBF}$ & $25.2 \pm 4.5$ & $120.4 \pm 10.1$ & $ \mathbf{0.2 \pm 0.0}$ \\ 
$\noise{\texttt{CBF}}$ & - & $25.0 \pm 4.4$ & $120.5 \pm 10.2$ & $\mathbf{0.3 \pm 0.1}$ \\ 

    \hline
    \end{tabular}
    \caption{Data set names and running times for $\gudhi$, $\dtm{\gudhi}$ and $\RipsNet_{\rm ucr}$ on UCR data sets.}
    \label{tab:ucr_time_long}
\end{table}

\begin{table*}[]
    \centering
    \begin{tabular}{|r|ccc|c|}
    \hline
        $\lambda$ ($\%$)                        & $\class_{\gudhi}^{\scriptscriptstyle{\mathrm{NN}}}$ & $\class_{\dtm{\gudhi}}^{\scriptscriptstyle{\mathrm{NN}}}$     & $\class_{\RipsNet}^{\scriptscriptstyle{\mathrm{NN}}}$  & \texttt{pointnet} \\ 
    \hline
        $0$         & $30.4 \pm 4.0$    & $30.9 \pm 2.0$              & $\color{red} 53.9 \pm 2.4 $           & $\mathbf{81.6\pm 1.1}$\\
        $2$      & $30.3 \pm 3.2$    & $31.0 \pm 2.7$              & $\color{red} 53.2 \pm 2.5 $           & $\mathbf{74.5\pm 1.6}$ \\
        $5$      & $29.9 \pm 4.0$    & $31.0 \pm 2.7$              & $\color{red} 55.1 \pm 3.3$            & $\mathbf{63.4\pm 1.6}$ \\
        $10$       & $25.2 \pm 3.2$    & $29.5 \pm 3.1$              & $\color{red} \mathbf{51.0 \pm 2.1}$   & $50.6\pm 1.5$ \\
        $15$      & $22.9 \pm 4.6$    & $25.7 \pm 3.1$              & $\color{red} \mathbf{46.9 \pm 3.0}$   & $44.9\pm 1.7$ \\
        $25$      & $14.4 \pm 4.0$    & $18.1 \pm 2.6$              & $\color{red} \mathbf{42.6 \pm 2.5}$   & $11.0\pm 0.2$ \\
        $50$       & $14.0 \pm 3.4$    & $13.1 \pm 1.9$              & $\color{red} \mathbf{31.6 \pm 3.3} $  & $10.9\pm 0.0$ \\
        $75$      & $11.3 \pm 1.5$    & $11.2 \pm 2.0$              & $\color{red} \mathbf{17.0 \pm 2.3} $  & $10.9\pm 0.0$ \\
        $90$       & $11.0 \pm 2.4$    & $10.8 \pm 3.1$              & $\color{red} \mathbf{12.8 \pm 2.8} $  & $10.9\pm 0.0$ \\
        \hline
    \end{tabular}
    \caption{Accuracy scores of simple neural net classifiers of $\gudhi$ and $\RipsNet$ on $\modelnet{10}$.
    $\lambda$ is the noise fraction and $(y-x)^2$ was used as persistence image weight function.
    The highest accuracy of the three topology based models $\class_{\gudhi}^{\scriptscriptstyle{\mathrm{NN}}}$, $\class_{\dtm{\gudhi}}^{\scriptscriptstyle{\mathrm{NN}}}$ and $\class_{\RipsNet}^{\scriptscriptstyle{\mathrm{NN}}}$ is highlighted in red, and the highest accuracy over all models, including the \texttt{pointnet} baseline, is highlighted in bold font.}
    \label{tab:modelnet10_acc_NN_full}
\end{table*}

\begin{table*}[]
    \centering
    \begin{tabular}{|r|ccc|c|}
    \hline
        $\lambda$ ($\%$)       & $\class_{\gudhi}^{\scriptscriptstyle{\mathrm{XGB}}}$ & $\class_{\dtm{\gudhi}}^{\scriptscriptstyle{\mathrm{XGB}}}$     & $\class_{\RipsNet}^{\scriptscriptstyle{\mathrm{XGB}}}$ & \texttt{pointnet} \\ 
    \hline
        $0$         & $32.2 \pm 2.8$    & $31.6 \pm 2.0$              & $\color{red} 49.1\pm 2.2 $  & $\mathbf{81.6\pm 1.1}$ \\
        $2$      & $31.0 \pm 4.9$    & $30.9 \pm 2.8$              & $\color{red} 48.3\pm 3.0 $  & $\mathbf{74.5\pm 1.6}$ \\
        $5$      & $30.4 \pm 2.6$    & $30.9 \pm 3.0$              & $\color{red} 48.0\pm 3.2 $  & $\mathbf{63.4\pm 1.6}$ \\
        $10$       & $28.3 \pm 2.0$    & $27.6 \pm 2.0$              & $\color{red} 46.0\pm 2.2 $  & $\mathbf{50.6\pm 1.5}$ \\
        $15$      & $26.6 \pm 2.8$    & $28.2 \pm 2.6$              & $\color{red} 43.3\pm 2.7 $  & $\mathbf{44.9\pm 1.7}$ \\
        $25$      & $21.6 \pm 2.9 $   & $25.6 \pm 2.0$              & $\color{red} \mathbf{40.7 \pm 2.8}$ & $        11.0\pm 0.2 $ \\
        $50$       & $15.3 \pm 2.0$    & $15.7 \pm 1.9$              & $\mathbf{\color{red} 27.8\pm 2.7} $  & $10.9\pm 0.0$ \\
        $75$      & $12.8 \pm 1.5$    & $11.9 \pm 1.1$              & $\mathbf{\color{red} 19.4\pm 1.6} $  & $10.9\pm 0.0$ \\
        $90$       & $13.0 \pm 2.1$    & $11.1 \pm 0.9$              & $\mathbf{\color{red} 13.1\pm 2.4} $  & $10.9\pm 0.0$ \\
        \hline
    \end{tabular}
    \caption{Accuracy scores of XGBoost classifiers of $\gudhi$ and $\RipsNet$ on $\modelnet{10}$.
    $\lambda$ is the noise fraction and $(y-x)^2$ was used as persistence image weight function.
    The highest accuracy of the three topology based models $\class_{\gudhi}^{\scriptscriptstyle{\mathrm{XGB}}}$, $\class_{\dtm{\gudhi}}^{\scriptscriptstyle{\mathrm{XGB}}}$ and $\class_{\RipsNet}^{\scriptscriptstyle{\mathrm{XGB}}}$ is highlighted in red, and the highest accuracy over all models, including the \texttt{pointnet} baseline, is highlighted in bold font.}
    \label{tab:modelnet10_acc_xgb}
\end{table*}


\begin{table}[]
    \centering
    \begin{tabular}{|r|ccc|}
        \hline
        $\lambda$ ($\%$) & $\gudhi$ (s)  & $\dtm{\gudhi}$ (s) & $\RipsNet$ (s)\\
        \hline
        $2$ & $118.4\pm4.7$     & $178.5\pm8.1$     & $\mathbf{0.2 \pm 0.0}$   \\
        $5$ & $117.8\pm4.5$     & $180.0\pm9.2$     & $\mathbf{0.2 \pm 0.0}$   \\
        $10$  & $117.5\pm4.6$     & $181.9\pm8.1$     & $\mathbf{0.2 \pm 0.0}$   \\
        $15$ & $120.0\pm4.7$     & $178.7\pm8.4$     & $\mathbf{0.3 \pm 0.0}$   \\
        $25$ & $121.2\pm4.4$     & $179.8\pm 7.8$    & $\mathbf{0.2 \pm 0.0}$  \\
        $50$  & $127.0\pm 6.4$    & $196.5\pm 10.7$    & $\mathbf{0.2 \pm 0.0}$   \\
        \hline
    \end{tabular}
    \caption{Running times on \modelnet{10} data in seconds, for $\gudhi$, $\dtm{\gudhi}$ and $\RipsNet$, respectively.}
    \label{tab:modelnet10_time}
\end{table}


\begin{sidewaysfigure}
    \centering
    \includegraphics[width=0.18\textwidth]{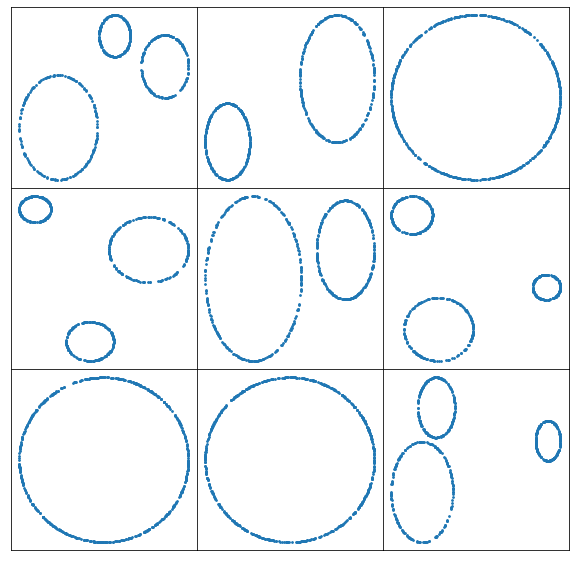}
    \includegraphics[width=0.18\textwidth]{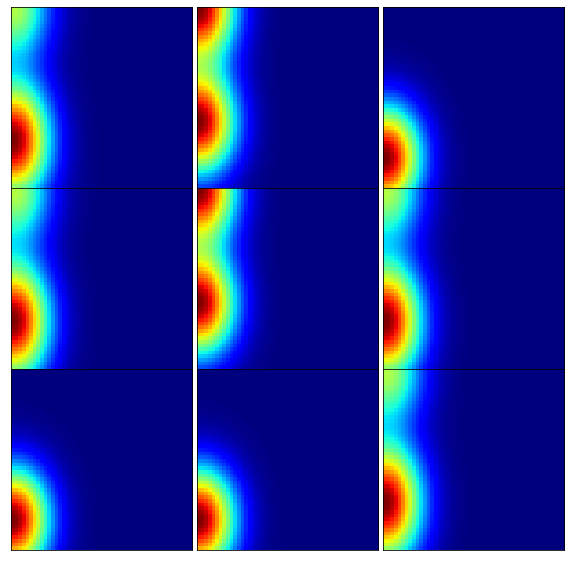}
    \includegraphics[width=0.18\textwidth]{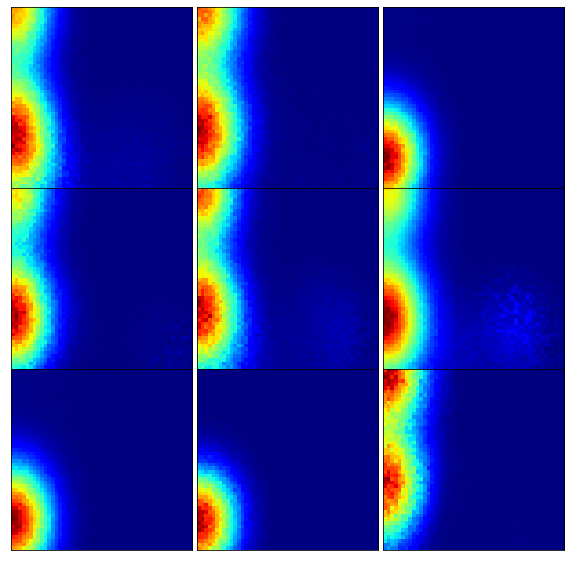}
    \includegraphics[width=0.18\textwidth]{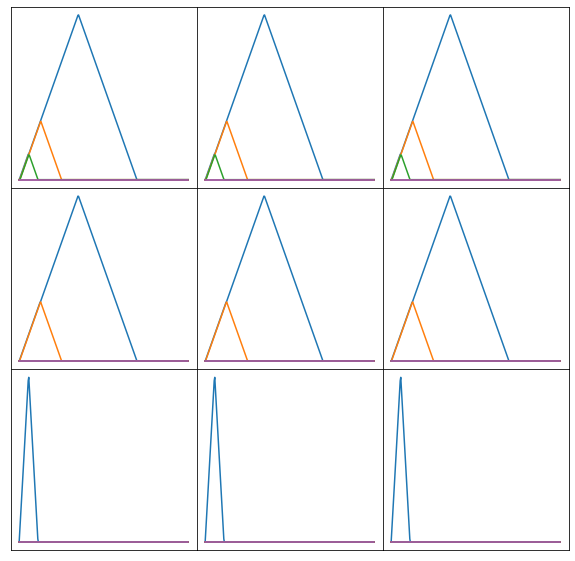}
    \includegraphics[width=0.18\textwidth]{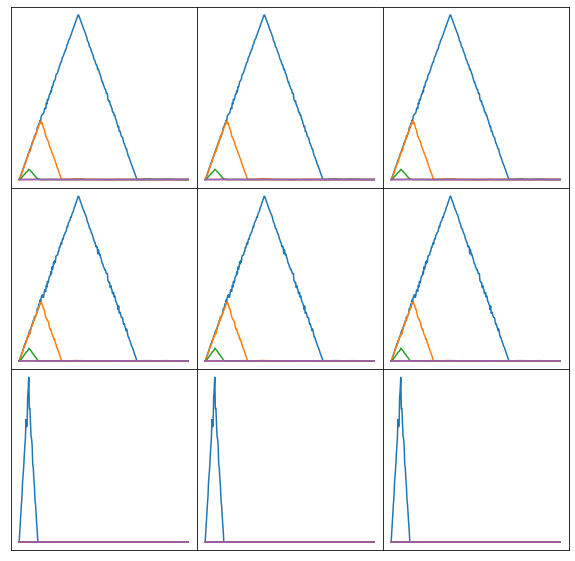}
    
    \vspace{1cm}
    
    \includegraphics[width=0.18\textwidth]{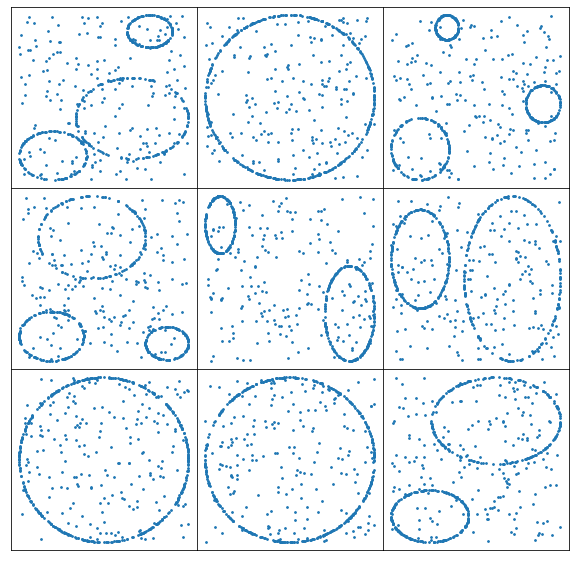}
    \includegraphics[width=0.18\textwidth]{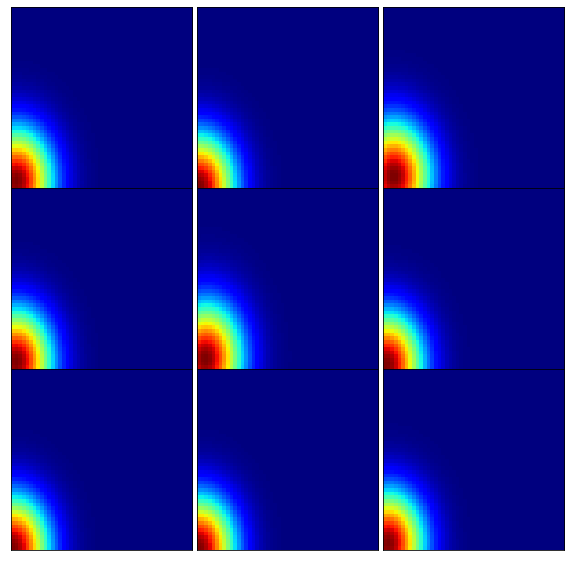}
    \includegraphics[width=0.18\textwidth]{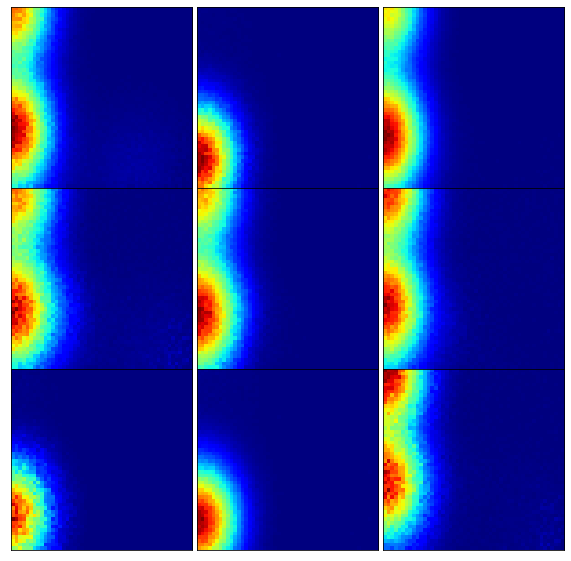}
    \includegraphics[width=0.18\textwidth]{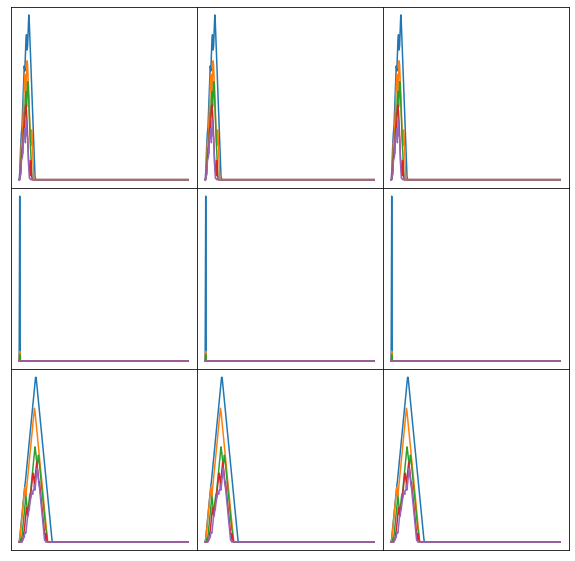}
    \includegraphics[width=0.18\textwidth]{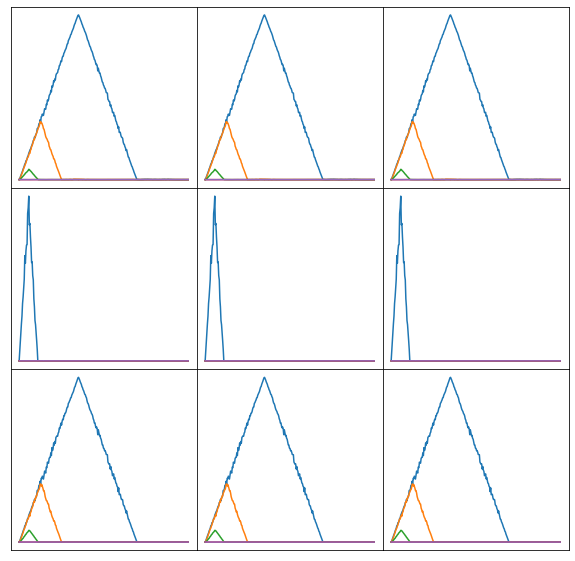}
    \caption{Larger version of \cref{fig:synth}}
    \label{fig:synth_large}
\end{sidewaysfigure}


\end{document}